\documentclass[aps,prx,secnumarabic,nofootinbib]{revtex4-2}

\usepackage{paper-style}

\newcommand{\TEpsilon}{\tilde {\mathcal E}}
\newcommand{\nbs}{\nobreakspace{}}
\newcommand{\TEpsilonQI}{\tilde {\mathcal E}_{\rm QI}}

\newcommand{\hilbert}{\mathcal{H}}

\newcommand{\Tr}{\operatorname{Tr}}
\newcommand{\Trivmap}{\operatorname{T}}

\newcommand{\divG}{\frac{1}{\text{Vol}(G)}}

\newcommand{\volg}{\text{Vol}(G)}

\usepackage{bbm} 

\def\qed{$\,\blacksquare$\par}

\def\Tr{\operatorname{Tr}}

\def\T+{\mathsf{T}_+}

\usepackage{stackengine}
\usepackage[normalem]{ulem}
\DeclareSymbolFont{sfletters}{OML}{cmbrm}{m}{it}
\DeclareMathSymbol{\srho}{\mathord}{sfletters}{"1A}
\DeclareMathSymbol{\salpha}{\mathord}{sfletters}{"0B}
\DeclareMathSymbol{\sbeta}{\mathord}{sfletters}{"0C}
\DeclareMathSymbol{\sgamma}{\mathord}{sfletters}{"0D}
\DeclareMathSymbol{\sdelta}{\mathord}{sfletters}{"0E}
\DeclareMathSymbol{\sepsilon}{\mathord}{sfletters}{"0F}
\DeclareMathSymbol{\szeta}{\mathord}{sfletters}{"10}
\DeclareMathSymbol{\seta}{\mathord}{sfletters}{"11}
\DeclareMathSymbol{\stheta}{\mathord}{sfletters}{"12}
\DeclareMathSymbol{\siota}{\mathord}{sfletters}{"13}
\DeclareMathSymbol{\skappa}{\mathord}{sfletters}{"14}
\DeclareMathSymbol{\slambda}{\mathord}{sfletters}{"15}
\DeclareMathSymbol{\smu}{\mathord}{sfletters}{"16}
\DeclareMathSymbol{\snu}{\mathord}{sfletters}{"17}
\DeclareMathSymbol{\sxi}{\mathord}{sfletters}{"18}
\DeclareMathSymbol{\spi}{\mathord}{sfletters}{"19}
\DeclareMathSymbol{\srho}{\mathord}{sfletters}{"1A}
\DeclareMathSymbol{\ssigma}{\mathord}{sfletters}{"1B}
\DeclareMathSymbol{\stau}{\mathord}{sfletters}{"1C}
\DeclareMathSymbol{\supsilon}{\mathord}{sfletters}{"1D}
\DeclareMathSymbol{\sphi}{\mathord}{sfletters}{"1E}
\DeclareMathSymbol{\schi}{\mathord}{sfletters}{"1F}
\DeclareMathSymbol{\spsi}{\mathord}{sfletters}{"20}
\DeclareMathSymbol{\somega}{\mathord}{sfletters}{"21}
\DeclareMathSymbol{\svarepsilon}{\mathord}{sfletters}{"22}
\DeclareMathSymbol{\svartheta}{\mathord}{sfletters}{"23}
\DeclareMathSymbol{\svarpi}{\mathord}{sfletters}{"24}
\DeclareMathSymbol{\svarrho}{\mathord}{sfletters}{"25}
\DeclareMathSymbol{\svarsigma}{\mathord}{sfletters}{"26}
\DeclareMathSymbol{\svarphi}{\mathord}{sfletters}{"27}

\newcommand\Item[1][]{%
  \ifx\relax#1\relax  \item \else \item[#1] \fi
  \abovedisplayskip=0pt\abovedisplayshortskip=0pt~\vspace*{-\baselineskip}}

\begin{document}
\title{Frame-Dependent Traces and the Third-Particle Paradox}
\author{Alessandro Palumbo}
\email{apalumbo@student.ethz.ch}
\affiliation{Institute for Theoretical Physics, ETH Zürich, 8093 Zürich, Switzerland}
\author{Luca Apadula}

\email{luca.apadula@telecom-paris.fr}
\affiliation{LTCI, Inria, Télécom Paris, Institut Polytechnique de Paris, 91120 Palaiseau, France}

\begin{abstract}

The Paradox of the Third Particle arises when comparing subsystem descriptions across Quantum Reference Frame (QRF) perspectives. We isolate two distinct origins of the Paradox: the QRF covariance of the partial trace and the failure of the physical Hilbert space to inherit the kinematical tensor-product structure. We give an explicit counterexample to the Relational Trace (RT) resolution: an uncorrelated product state for which the RT statistical condition trivialises. We then introduce a new statistical consistency condition comparing subsystem discarding between external and internal QRFs, together with an associated frame-dependent map, the Perspective Relational Trace (PRT). We argue that our condition captures the operational content of the Paradox: rather than imposing consistency on the whole state space, we characterise exactly the states on which it holds in the Perspective-Neutral (PN) and Quantum-Information (QI) approaches. This separates three levels of description: a PN subsystem of a PN whole, where consistency fails on a characterised set that includes product states; a QI subsystem of a QI whole, where it holds for all states; and a QI subsystem obtained from a PN whole by kinematical partial trace, where the full weakly invariant algebra is recovered, yet consistency holds only on a proper subset. These results show that the PN approach can consistently describe only a closed, isolated system, while the QI approach can accommodate arbitrary subsystems. Tracing out a subsystem from a globally PN state yields a charge-superselected algebra, reproducing in a minimal QRF model the boundary-charge structure of edge modes. We understand the Paradox not as a genuine contradiction, but as the consequence of comparing inequivalent physical layers without tracking which information is externally and which internally accessible.

\end{abstract}

\maketitle

\section{Introduction}

The formalism of Quantum Reference Frames (QRFs) has received increasing attention in recent years from the quantum foundations\nbs\cite{Giacomini2019,Angelo_2011,paradoxmuller,Vanrietvelde2020ChangePerspective,vantrietvelde_switching_2018,Giacomini2019RelativisticQRF,CastroRuiz2020QuantumClocks,deLaHamette2020QRF,CastroRuiz2025RelativeSubsystems,delahamette2021perspectiveneutral,delahamette2021entanglementasymmetry,delahamette2022quantum,Apadula2024QuantumReferenceFramesLorentz,Doat2025SymmetryConstrained,garmier2025perspectivesnonidealquantumreference,Apadula2026FramePerspectives,delahamette2025quantumreferenceframesarbitrary,ballesteros_2020,cepollaro2024}, quantum information\nbs\cite{spekkens,loveridge_symmetry_2018,Loveridge2017RelativityObservables,Carette2025operationalquantum,deLaHamette2026ObserverDependentEntropy,Gour2008ResourceTheory,Palmer2014ChangingQRF,Smith2016NoncompactGroups,poulin_dynamics_2007,Miyadera2016AccuracyTradeoff,glowacki,Giacomini2021SpacetimeQRF,paololuppi}, and quantum gravity\nbs\cite{Hoehn2021TrinityRelational,hoehn2021quantum,mikusch2021transformation,hoehn2021internal,Cepollaro2021QuantumEquivalence,delaHamette2021falling,Kabel2022ConformalSymmetries,kabel2023quantum,Hoehn_2023_subsystems,Kabel2025QuantumCoordinates,Chen2026QuantumReferenceFields,hoehn_switching_2019,Hoehn2021EquivalenceRelational,Streiter_2020,devuyst2024gravitationalentropyobserverdependent,fewster2024,DeVuyst2025CrossedProductsQRF,Rothlin2026ErrorCorrectionLatticeQED,Lacambra2026GaussLawCodes,AraujoRegado:2025relational} communities. The main idea of QRFs is to extend the notion of reference frames to the quantum domain\nbs\cite{Giacomini2019}, allowing for transformations between different QRF perspectives. The framework of QRFs has allowed us to rethink observables, subsystems, and entropies in gravity and gauge theories in a relational and gauge-invariant fashion\nbs\cite{Hoehn_2023_subsystems,hoehn2021quantum,AraujoRegado:2025relational,Hoehn2021TrinityRelational,devuyst2024gravitationalentropyobserverdependent,DeVuyst2025CrossedProductsQRF,fewster2024,deLaHamette2026ObserverDependentEntropy}. Different approaches to QRFs have been developed, such as the Perspective-Neutral (PN)\nbs\cite{vantrietvelde_switching_2018,Vanrietvelde2020ChangePerspective,delahamette2021perspectiveneutral,Hoehn2021TrinityRelational,hoehn2021internal}, Extra-Particle\nbs\cite{CastroRuiz2025RelativeSubsystems,garmier2025perspectivesnonidealquantumreference}, Operational\nbs\cite{Miyadera2016AccuracyTradeoff,glowacki2024quantumreferenceframeshomogeneous,Loveridge2017RelativityObservables,loveridge_symmetry_2018}, Algebraic\nbs\cite{Bojowald2021QuantizationDynamicalSymplecticReduction,Bojowald2023AlgebraicPropertiesQRF,Bojowald2023FrozenFormalismProblemOfTime}, and Effective\nbs\cite{Bojowald2009EffectiveConstraintsQuantumSystems,Bojowald2009EffectiveConstraintsRelativisticQuantumSystems,Bojowald2011EffectiveApproachProblemTime,Bojowald2011EffectiveApproachProblemTimeGeneralFeatures,Hoehn2012EffectiveRelationalDynamicsNonintegrableCosmologicalModel} approaches. These frameworks rest on different underlying assumptions and address different physical questions, and therefore yield different physical results\nbs\cite{castroruiz2025interpretingquantumreferenceframe,DeVuyst2025PerspectiveNeutralAlgebraicEffectiveQRF}.
In this work we discuss a problem known as the \emph{Third-Particle Paradox}, first introduced in Ref.\nbs\cite{Angelo_2011}. The authors of Ref.\nbs\cite{Angelo_2011} construct a scenario in which the following natural physical intuition appears to fail: a system that is uncorrelated with the degrees of freedom one cares about should have no bearing on the information accessible about those degrees of freedom. One starts with two particles and encodes some information in their relational degrees of freedom, that is, in the configuration of one particle relative to the other. From the perspective of one of the two particles, which we take to be particle~1, this information is accessible. One then enlarges the description by introducing a third particle, prepared in a state uncorrelated with the first two and possibly far away from them. We would expect this addition to be innocuous. Yet, when the description of the relative configuration of particles~1 and~2 is reconstructed in the presence of the third particle, the information that was previously accessible to particle~1 is found to be lost. The inclusion of an uncorrelated particle thus appears to affect what an observer can learn about a subsystem with which that particle has nothing to do. This is troubling because it strikes at the consistency of the QRFs formalism itself: if whether or not one accounts for an irrelevant particle changes the physics accessible from a given QRF, then the very notion of a frame-relative description of a subsystem is called into question.

The paradoxical situation introduced in Ref.\nbs\cite{Angelo_2011} can be naturally framed, in a conceptually equivalent scenario, within the PN approach to QRFs. The Paradox was subsequently analysed by the authors of Ref.\nbs\cite{paradoxmuller}, who propose an alternative resolution via the so-called Relational Trace (RT), a notion of subsystem-discarding defined directly at the level of relational degrees of freedom, and that seems to dissolve the source underlying the Paradox. In the present work, we exhibit an explicit counterexample to this resolution, where the RT condition becomes trivial, showing that satisfying the RT condition does not by itself certify an operational resolution of the Paradox. This leaves open a set of questions that motivate our analysis: what is the true origin of the Paradox, whether it signals a genuine defect of the QRF formalism or instead a mismatch in how subsystems are compared across perspectives, and which notion of subsystem discarding---if any---resolves it.

We start our contributions by analysing precisely why, and how, the Paradox emerges. We identify two conceptually distinct mechanisms behind it. The first is that, when we discard the third uncorrelated particle, we are not accounting for the QRF covariance of quantum channels---and hence the partial trace. We resolve this first problem straightforwardly by introducing the \emph{Perspective Trace}, a discarding operation that transforms covariantly under QRF transformations. The second, more subtle facet is specific to the PN framework: the projection onto the physical Hilbert space does not inherit the kinematical TPS, so that tracing out a subsystem before and after physicalisation are, in general, inequivalent operations, and no standard notion of partial trace exists on the physical subspace. At first sight, this second problem can be resolved via the RT\nbs\cite{paradoxmuller}. However, as discussed above, this discarding operation does not seem to capture the operational content of the Paradox, and fails on a class of states.

We develop a new framework by reformulating the operational content of the Paradox via a novel statistical consistency condition, together with a corresponding notion of subsystem discarding, the \emph{Perspective Relational Trace}. A distinctive feature of our approach is that we do not require the condition to hold on the whole state space; instead, we characterise exactly the states on which it can be satisfied. We argue that the fact that our condition is not satisfied for all states in the PN approach is not a shortcoming of our construction but its \emph{key} feature: it singles out precisely those configurations for which the subsystem relational information accessible to an \emph{external} frame---one defined prior to projection onto the physical Hilbert space---cannot be recovered \emph{internally}, i.e.\ from a QRF with access only to relational degrees of freedom, once the third particle is discarded internally.

Our analysis separates three levels at which subsystem~12 can be described once particle~3 is discarded. \emph{(i)} When both the total system and the subsystem are described in the PN approach, the resolution fails on a characterised set of states, including product states for which particle~3 is uncorrelated with subsystem~12; the PN approach thus consistently describes a closed system---the \enquote{whole universe}---but cannot, in general, reduce to a consistent subsystem description. \emph{(ii)} When the QI (Extra-Particle) description is adopted from the outset, for both the total system and the subsystem, the condition holds for \emph{all} states: the relational information about~12 accessible externally, to Eve, coincides with that recovered internally after discarding particle~3. \emph{(iii)} In the intermediate case, the total system is described in the PN approach and the subsystem~12 is reached through the kinematical partial trace; here we make precise the sense in which a QI description \emph{emerges} from a PN one, as at the level of operator algebras the partial trace maps the physical ideal onto the entire weakly invariant ideal (the trace-class operator space of the QI approach). This emergence, however, is not faithful state by state: there remain states whose externally accessible subsystem relational information is not recovered internally after tracing, so that an instance of the Paradox persists. Our findings thus highlight that the weakly invariant description is the appropriate one for proper subsystems, whereas the PN approach is not stable under subsystem discarding, and can only describe the \enquote{whole universe}, in line with the discussion of Ref.\nbs\cite{CastroRuiz2025RelativeSubsystems}. 

Finally, we read the Third-Particle Paradox in the PN approach as a toy model providing a concrete QRF realisation of the argument of Ref.\nbs\cite{Rovelli2014WhyGauge}. Therefore, we argue that this mechanism underlying the Paradox also persists at the classical level, complementing the classical perspectival formulation recently established in Ref.~\cite{rennerhausmann}. Furthermore, we identify the QI emergence from the kinematical partial trace of a PN state as the analogue, in a minimal QRF setting, of the boundary-charge superselection structure familiar from the edge-mode program\nbs\cite{Donnelly:2011hn,Donnelly2014nonabelian,Donnelly_2016,CasiniHuertaRosabal2014,Soni:2015yga,VanAcoleyen:2015ccp,Delcamp:2016eya,Carrozza2024EdgeModesDynamicalFrames,AraujoRegado2025SoftEdges,AraujoRegado:2025relational}. 

We conclude by arguing that the Paradox is not a genuine contradiction, but the consequence of comparing distinct, inequivalent layers of physical description. Our analysis also clarifies in which physical scenarios each of the two QRF frameworks is more appropriate.

\par\medskip
\noindent\textbf{Note.}\quad
The present paper is mostly based on A.P.'s Semester Project at ETH Zürich  (October 2025)~\cite{Palumbo2026FrameDependentTraces}.

\section{Quantum Reference Frames and the Covariance of Quantum Channels}\label{sec:qrf_frameworks}
In this section, we provide an overview of the main approaches to QRFs and introduce the notation used in the following sections. We consider a physical scenario with an associated symmetry\footnote{The group $G$ can be associated with either a gauge or a physical symmetry, as discussed below.} group $G$. We assume $G$ to be a compact Lie group. We describe $n$ particles, whose corresponding kinematical Hilbert space reads
\begin{equation}
    \hilbert \;=\; \bigotimes_{i=1}^{n} \hilbert_i,
    \qquad \text{with } \ \hilbert_i = L^2(G).
\end{equation}
Each of the particles is therefore a complete and ideal QRF. This means that for each frame there exists a coherent state system $\{\ket{g}_i\}_{g\in G}$ on which $G$ acts both transitively and freely\nbs\cite{delahamette2021perspectiveneutral}. Furthermore, we assume that the reference frame associated with particle $i$ transforms under the left-regular representation of $G$, denoted by $U_i(g)$. Since the QRFs are ideal, the coherent states are perfectly distinguishable $
    \braket{g|g'}_i \;=\; \delta(gg'^{-1})$. The left action of $G$ on the coherent state system is given by $  U_i(g)\,\ket{g'}_i \;=\; \ket{gg'}_i$. 
The left-regular representation acts on wavefunctions $\psi \in L^2(G)$ as
$(U_i(g)\psi)(g') \;=\; \psi(g^{-1}g')$. 
We note that for continuous Lie groups $G$, the kets $\ket{g}$ are not elements of $L^2(G)$;\footnote{For continuous $G$, the kets $\{\ket{g}_i\}_{g\in G}$ are distributions in the sense of rigged Hilbert spaces\nbs\cite{delahamette2021perspectiveneutral,schwartz1950theorie,gelfand2016generalized}.} one typically works in the wavefunction representation. In what follows, however, we will nevertheless act directly on $\ket{g}$ to treat continuous and discrete cases simultaneously. The Peter--Weyl theorem allows us to decompose $\hilbert_i$ as 
\begin{equation} \label{eq:HS_decomposition}
    \hilbert_i = \bigoplus_{q\in \hat G} \hilbert_{i_L}^{(q)} \otimes \hilbert_{i_R}^{(q)}
\end{equation}
where $\hat G$ denotes the set of irreps of the group, labelled by $q$. The kets $\left\{\ket{g}_i\right\}_{g \in G}$ admit a decomposition in an orthonormal basis of irreducible representations of $G$~\cite{PhysRevA.69.052326} $
\ket{g}_i = \sum_{q,m,n} \sqrt{\frac{d_q}{\volg}} D^{(q)}_{mn} (g) \ket{q,m,n}_i$ 
where $\left\{\ket{q,m,n}_i\right\}$ is an orthonormal basis for $L^2(G)$, $d_q$ is the dimension of the irrep $q$, $\text{Vol}(G) = \int_G dg$ and $D^{(q)}_{mn}(g)$ are the matrix elements of $U_i^{(q)}(g)$, corresponding to the irreducible sectors in the decomposition of $U_i(g)$. The left-regular representation acts on the gauge index $m$, whereas the right-regular representation acts on the multiplicity index $n$. 
Furthermore, the coherent state system $\left\{\ket{g}_i\right\}_{g \in G}$ provides a resolution of the identity \cite{delahamette2021perspectiveneutral}. Hence, any bounded operator $T_i \in\mathcal{B}(\hilbert_i)$ can be represented in the group basis by its kernel $\bra{g}T_i\ket{g'}$, understood in the distributional sense for continuous $G$. States are positive semidefinite, unit-trace, trace-class operators on $\hilbert_i$, and therefore form a convex subset of $\mathcal{B}_1(\hilbert_i)$.\footnote{Since observables are bounded and states are trace-class, products of the form $A\rho$ are trace-class and expectation values $\Tr(A\rho)$ are finite.}

In this paper, we adopt the terminology introduced in Ref.~\cite{Doat2025SymmetryConstrained}: the coherent group averaging is referred to as \emph{strong twirling}, while the incoherent group averaging ($G$-twirl) is referred to as \emph{weak twirling}.

\subsection{Perspective neutral approach: strong twirling}\label{sec:framework_PN}

We begin by reviewing the Perspective-Neutral (PN) approach to quantum frame covariance
\cite{Vanrietvelde2020ChangePerspective, vantrietvelde_switching_2018, delahamette2021perspectiveneutral, Hoehn2021TrinityRelational,  hoehn2021internal}. In this framework, the starting point is the perspective of an external, possibly fictitious, reference frame, denoted here by Eve. This frame is assumed to encode a global external description of the $n$ particles, whereas each particle only has access to the degrees of freedom of the other particles relative to itself. Crucially, the particles do not have access to Eve and cannot reconstruct Eve's external view, but only relational degrees of freedom between the particles. The particles are therefore referred to as \emph{internal frames}, whereas Eve plays the role of an \emph{external frame}. Within the PN approach, Eve is typically regarded as an auxiliary, unphysical perspective. Accordingly, this external reference frame is interpreted as pure \emph{gauge}, and only relational degrees of freedom in the trivial-charge sector of the gauge group are taken to be physical. Formally, this corresponds to projecting the kinematical state space onto the zero-charge sector of the gauge group.

The perspective of Eve defines the kinematical Hilbert space, while the relational degrees of freedom between the particles are encoded in the physical Hilbert space. A kinematical vector $\ket{\psi}\in\hilbert$ is mapped to its physical description via the coherent group averaging, or strong $G$-twirling, through the map $\Pi:\hilbert\mapsto\hilbert$ defined as
\begin{equation}
    \Pi = \divG\int_G dg \, U(g).
\end{equation}
For compact Lie groups, this defines an orthogonal projector onto the strongly invariant (physical) subspace. We define  $\hilbert^{\rm phys}:=\Pi\hilbert$. Indeed, the finiteness of the Haar volume, together with the left- and right-invariance of the Haar measure, implies $\Pi^2=\Pi$ and $\Pi^\dagger=\Pi$.\footnote{Thus, for compact $G$, the physical Hilbert space is a closed subspace of the kinematical Hilbert space. For non-compact groups, this does not hold: coherent group averaging generally fails to define a bounded projector on $\hilbert$, and the physical Hilbert space must instead be constructed through procedures such as RAQ or co-invariants\nbs\cite{delahamette2021perspectiveneutral,DeVuyst2025CrossedProductsQRF}.}

At the operator level, the algebra of trivial-charge operators is represented on the kinematical Hilbert space as\footnote{Since $\Pi\in\mathcal{B}(\hilbert)$, the compression $\Pi\mathcal{B}(\hilbert)\Pi$ is an algebra with unit $\Pi$ isomorphic to $\mathcal{B}(\hilbert^{\rm phys})$.}
\begin{equation}
    \mathcal{B}(\hilbert^{\rm phys}) 
    \;\simeq\; 
    \Pi\mathcal{B}(\hilbert)\Pi
    \;=\; 
    \bigl\{\,T \in \mathcal{B}(\hilbert) \;\big\vert\; U(g)T=T U(g)=T \quad\forall\,g \in G \bigr\},
\end{equation}
and we denote operators in this algebra as \emph{strongly invariant}.\footnote{One could also work with strongly invariant states and weakly invariant observables (defined in the next subsection). This does not change expectation values \cite{delahamette2021perspectiveneutral}.} It is then convenient to introduce the compression map $\Phi:\mathcal{B}(\hilbert)\to \Pi\mathcal{B}(\hilbert)\Pi$
\begin{equation}
    \Phi[T] := \Pi T \Pi.
\end{equation}
This map projects kinematical operators onto the physical sector.\footnote{Since $\Pi\mathcal{B}(\hilbert)\Pi$ is isomorphic to $\mathcal{B}(\hilbert^{\rm phys})$, the compressed operator may equivalently be regarded as an operator on $\hilbert^{\rm phys}=\Pi\hilbert$.} On trace-class operators, the same compression defines a CP trace-non-increasing map $\mathcal{B}_1(\hilbert)\to\mathcal{B}_1(\hilbert^{\rm phys})$. We also introduce the normalised projection
\begin{equation}
    \widehat{\Phi}[\rho]
    :=
    \frac{\Pi\rho\Pi}{\Tr[\rho\Pi]}
\end{equation}
whenever $\Tr[\rho\Pi]\neq 0$. For later convenience, we extend this definition to the zero-weight case by setting $\widehat{\Phi}[\rho]:=0$ whenever $\Tr[\rho\Pi]=0$. Note that $\widehat{\Phi}$ is trace-preserving whenever its denominator is nonzero but is nonlinear and therefore is not a CP map.

We shall use the following notation to distinguish global projections from projections acting on individual kinematical tensor factors. Given a kinematical tensor factor $i$, we denote by $S|i$ the collection of all remaining kinematical tensor factors. We then define the projector onto the trivial-charge sector of subsystem $i$:
\begin{equation}
    \Pi_{i} := \divG\int_G dg \ \mathbb{I}_{S|i}\otimes U_i(g),
\end{equation}
together with the associated compression map
\begin{equation}
    \Phi^{(i)}[T] := \Pi_{i}T \Pi_{i}.
\end{equation}
Here $\Phi^{(i)}$ maps kinematical operators to operators supported on the trivial-charge sector of subsystem $i$. For states with nonzero weight in this sector, i.e. whenever $\Tr[\rho\Pi_i]\neq 0$, we also define the normalised projection
\begin{equation}
    \widehat{\Phi}^{(i)}[\rho]
    :=
    \frac{\Pi_i\rho\Pi_i}{\Tr[\rho\Pi_i]},
\end{equation}
and extend it to the zero-weight case by setting $\widehat{\Phi}^{(i)}[\rho]:=0$ whenever $\Tr[\rho\Pi_i]=0$. 

On the physical Hilbert space, one can map to the perspective of the frame associated with particle $i$ via the Schr\"odinger reduction map $\mathcal{R}_{S|i}^{(i)}(g):\hilbert^{\rm phys}\to \hilbert_{S|i}$, defined by \cite{delahamette2021perspectiveneutral}
\begin{equation}
    \mathcal{R}_{S|i}^{(i)}(g) := \sqrt{\text{Vol}(G)} \,\bra{g}_i \otimes \mathbb{I}_{S|i}.
\end{equation}
For complete ideal frames, this map defines a unitary isomorphism between $\hilbert^{\rm phys}$ and $\hilbert_{S|i}$\nbs\cite[Corollary 2]{delahamette2021perspectiveneutral}.\footnote{Generally, the reduction map is a unitary from the physical space to its image. When the frame is complete and ideal, the image of the reduction map is the full system space $\hilbert_{S|i}$.} In App.\nbs\ref{app:pn_details}, we provide further mathematical details of the PN approach. For ideal frames, the internal QRF transformations of the PN framework coincide with those of the perspectival approach\nbs\cite{Giacomini2019,deLaHamette2020QRF}, implemented by the unitary $U_{i\rightarrow j}:\hilbert_{S|i}\rightarrow\hilbert_{S|j}$ (see App.\nbs\ref{app:pn_details}).

\subsection{Quantum information approach: weak twirling}\label{sec:QI-approach}

We now discuss what is commonly referred to as the extra-particle approach to quantum reference frames \cite{CastroRuiz2025RelativeSubsystems}. In this work, we refer to it as the Quantum Information (QI) approach to QRFs and to its associated symmetry-averaging operation as weak twirling. In analogy with the PN framework, we start from a kinematical Hilbert space describing the perspective of an external reference frame Eve. Crucially, in the QI approach Eve is not interpreted as an unphysical gauge frame. Rather, she represents an external frame resource that is inaccessible to the internal reference frames associated with the $n$ particles.
As a consequence, the information accessible to the internal frames is encoded in the \emph{invariant algebra}
\begin{equation}
    \mathcal{B}(\hilbert)^G
    \;:=\; 
    \bigl\{\,T \in \mathcal{B}(\hilbert) \;\big\vert\; [\,U(g),\,T\,] = 0 \quad\forall\,g \in G \bigr\}.
\end{equation}
We denote operators in this algebra as \emph{weakly invariant}. Equivalently, $\mathcal{B}(\hilbert)^G$ is the image of the weak $G$-twirling map\footnote{One can easily prove that $\mathcal{B}(\hilbert)^G = \text{Im}(\mathcal{G})$.} $\mathcal{G}:\mathcal{B}(\hilbert)\rightarrow\mathcal{B}(\hilbert)^G$, defined by
\begin{equation}
    \mathcal{G}[T] \;=\; \divG\int_G \!dg\;U(g)\,T\,U(g)^\dagger.
\end{equation}
For compact $G$, this map is a unital completely positive idempotent map. On trace-class operators, the same formula defines a CPTP map on states.
This admits a clear information-theoretic interpretation: since the orientation of the external reference frame is inaccessible to the internal frames, one averages over all possible orientations. The weak $G$-twirl implements precisely this averaging procedure, and its image corresponds to the information accessible to the internal frames in the QI approach. We define the weak $G$-twirl acting on subsystem $i$ by
\begin{equation}
    \mathcal{G}^{(i)}[T]
    := \divG
    \int_G dg \,
    \left(\mathbb{I}_{S|i} \otimes U_i(g)\right)
    T
    \left(\mathbb{I}_{S|i} \otimes U_i(g)^\dagger\right).
\end{equation}

To describe weakly invariant operators from the perspective naturally associated with particle $i$, one can pass to the corresponding refactorisation by means of the trivialisation map. More precisely, let
\begin{equation}
\operatorname{T}_i^{(S|i)}
:=
\int_G dg \;
\ket{g}\bra{g}_{i} \otimes U_{S|i}^\dagger(g),
\end{equation}
and let $F_{C,i}$ denote the Hilbert-space isomorphism that relabels the reference-frame factor $i$ as $C$. We define the corresponding superoperator by
\begin{equation}
    \mathcal{V}_i^{(S|i)}[T]
    :=
    V_i^{(S|i)} T V_i^{(S|i)\dagger},
    \qquad
    V_i^{(S|i)} := F_{C,i}\Trivmap_i^{(S|i)} .
\end{equation}
A key feature of the QI approach emerges when taking the perspective of an internal frame. The refactorised Hilbert space associated with particle $i$ admits the decomposition
\begin{equation}
     \hilbert_{C,S|i} := V_i^{(S|i)}\hilbert  \cong 
    \left(\bigoplus_{q\in \hat G} \hilbert_{C_L}^{(q)} \otimes \hilbert_{C_R}^{(q)}\right) \otimes \hilbert_{S|i}.
\end{equation}
Since $V_i^{(S|i)}$ is unitary, conjugation by $V_i^{(S|i)}$ realises a $*$-isomorphism $
    \mathcal{V}_i^{(S|i)}:\mathcal{B}(\hilbert) \rightarrow \mathcal{B}(\hilbert_{C,S|i})$. 
In this representation, weak invariance is invariance under the residual left action of $G$, which in the refactorised Hilbert space acts only on the $C_L$ factor. Accordingly, any weakly invariant operator $T_{C,S|i}^{\rm inv} \in \mathcal{B}(\hilbert_{C,S|i})^G$ takes the form
\begin{equation} \label{eq:invariantoperatorAlice}
    T_{C,S|i}^{\rm inv}
    \;=\;
    \bigoplus_{q\in \hat G} \mathbb{I}_{C_L}^{(q)} \otimes T^{(q)}_{C_R,S|i}.
\end{equation}
The factor $C_L$ contains the gauge degrees of freedom, whereas the right factor $C_R$ is associated with the so-called extra-particle degree of freedom\nbs\cite{CastroRuiz2025RelativeSubsystems}. In short, the extra-particle is what is missing from relational observables\nbs\cite{CastroRuiz2025RelativeSubsystems,delahamette2021perspectiveneutral} to obtain the full G-invariant algebra.\footnote{A word of caution on terminology. In the QRF literature, \enquote{relational observables} usually denote operators of the form\nbs\eqref{eq:real_rel_obs}. Throughout this work, however, we use \enquote{relational} more broadly, for the full algebra accessible from the perspective of the internal frames: $\mathcal{B}(\hilbert^{\rm phys})$ in the PN approach and $\mathcal{B}(\hilbert)^G$ in the QI approach. In the PN approach the algebra accessible from the internal frames is weakly-isomorphic to relational observables in the strict sense\nbs\cite[Theorems~1 and~2]{delahamette2021perspectiveneutral}, whereas in the QI approach it is strictly larger---the difference being precisely the extra-particle.} As we shall see in Sec.\nbs\ref{sec:weak-twirling-PRT}, the extra-particle plays a crucial role in the discussion of the Paradox of the Third Particle. In App.\nbs\ref{app:qi_details}, we provide more details on the extra-particle and the QI framework.

\subsection{Comparison between the approaches}\label{sec:framework_comparison}

In this subsection, we highlight the main differences between the two approaches introduced above. The first concerns the information accessible from the perspective of the internal frames. In the PN approach, this information is selected by the strong twirling, whereas in the QI approach it is selected by the weak twirling. Clearly, the strongly invariant algebra is contained in the weakly invariant algebra. Fig.~\ref{fig:symmetry} illustrates this difference at the level of the information \emph{internally} accessible in the two frameworks by displaying the density matrices obtained through the PN and QI constructions.\footnote{Strictly speaking, by Schur's lemma the weak $G$-twirl also applies a fully depolarising channel within each irrep. For Abelian groups, the irreducible representations are one-dimensional, and therefore this depolarising channel is trivial. Hence, for Abelian groups, the figure faithfully represents the action of the weak $G$-twirl.}

The comparison can be sharpened by viewing the PN construction inside the QI refactorised representation. More precisely, we show that fixing the internal frame to the identity in the PN approach corresponds, inside the QI approach, to restricting the invariant structure to the trivial-charge block. We thus proceed by applying the QI refactorisation map to strongly invariant operators. For any kinematical operator $T$, projecting onto the physical sector and then applying the QI refactorisation associated with particle $i$ yields
\begin{equation}\label{eq:pn_extra_particle}
    T_{C,S|i}^{\rm PN} := \mathcal{V}_i^{(S|i)}\!\left[\Phi[T]\right]
    =
    P_0^C \otimes T_{S|i}^{(e)},
\end{equation}
where
\begin{equation}
    T_{S|i}^{(e)} := \mathcal{R}_{S|i}^{(i)}(e)\,\Phi[T]\,\mathcal{R}_{S|i}^{(i)\dagger}(e),
    \qquad P_0^C := \ket{q=0}\bra{q=0}_C, \qquad
    \ket{q=0}_C := \frac{1}{\sqrt{\text{Vol}(G)}}\int_G dg\,\ket{g}_C .
\end{equation}
Thus, the PN sector appears inside the QI refactorised description as the trivial-charge block. The $C$ factor is fixed to the trivial-charge rank-one projector $\ket{q=0}\bra{q=0}_C$ and remains uncorrelated with the description of $S$ from the perspective of frame $i$, as in Ref.~\cite{Vanrietvelde2020ChangePerspective}. In this restricted sector, the refactorisation map acts as a disentangler, in agreement with the terminology of Ref.~\cite{delahamette2021perspectiveneutral}.

The PN reduction can then be obtained by applying the Schr\"odinger reduction map to the refactorised operator in Eq.~\eqref{eq:pn_extra_particle}. Conditioning $C$ on any orientation $g\in G$ gives the same reduced relational operator as the one obtained via the PN prescription:
\begin{equation}\label{eq:from_EP_to_PN}
    T_{S|i}^{(e)}
    =
    \mathcal{R}^{(C)}_{S|i}(g)\,T_{C,S|i}^{\rm PN}\,\mathcal{R}^{(C)\dagger}_{S|i}(g)
    =
    \mathcal{R}^{(i)}_{S|i}(e)\,\Phi\left[T\right]\,\mathcal{R}^{(i)\dagger}_{S|i}(e).
\end{equation}
This independence of $g$ is the analogue, in the QI refactorised representation, of the corresponding property underlying the Heisenberg reduction map of Ref.~\cite{delahamette2021perspectiveneutral}. This construction shows how the strongly invariant sector appears when viewed through the weakly invariant, refactorised description. More explicitly, in the refactorised representation, the strong projection can be viewed as the weak \(G\)-twirl followed by restriction to the trivial charge block of \(C\). Indeed, after refactorisation one has $V_i^{(S|i)}\Pi V_i^{(S|i)\dagger}=P_0^C\otimes\mathbb{I}_{S|i}$, and therefore
\begin{equation}
    V_i^{(S|i)}\Phi[T]V_i^{(S|i)\dagger} = (P_0^C\otimes\mathbb{I}_{S|i})\,\mathcal{G}_C\!\left[V_i^{(S|i)}T V_i^{(S|i)\dagger}\right]\,(P_0^C\otimes\mathbb{I}_{S|i}),
\end{equation}
where $\mathcal{G}_C
    :=
    \mathcal{V}_i^{(S|i)}\circ \mathcal{G}\circ
    \left(\mathcal{V}_i^{(S|i)}\right)^{-1}$ denotes the weak $G$-twirl in the refactorised representation.

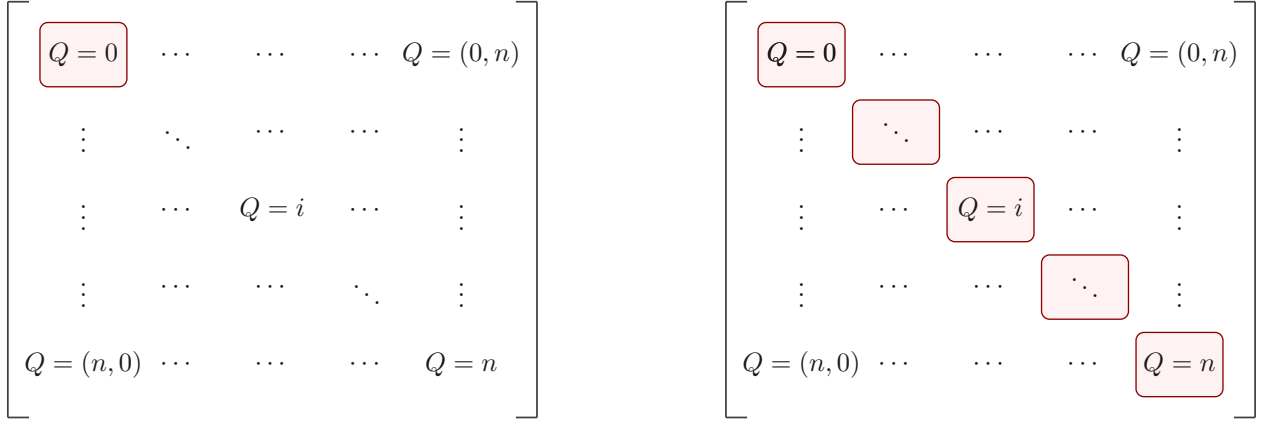
\begin{figure}[ht]
\centering
\begin{tikzpicture}[
    font=\normalsize,
    charge/.style={
    text=black!90
    },
    hblock/.style={
    draw=red!55!black,
    fill=red!5,
    line width=0.5pt,
    rounded corners=3pt
    },
    bracket/.style={
    line width=0.65pt,
    draw=black!75
    }
]

\def\W{7}
\def\H{5.5}
\def\gap{2.5}

\def\bTop{4.95}
\def\bBot{0.45}

\def\xA{1}
\def\xB{2.25}
\def\xC{3.50}
\def\xD{4.75}
\def\xE{6}

\def\yA{4.8}
\def\yB{3.775}
\def\yC{2.75}
\def\yD{1.725}
\def\yE{0.7}

\newcommand{\drawmatrixbrackets}[2]{%
    \draw[bracket] ($(#1,#2)+(0,\H)$) -- ($(#1,#2)+(0,0)$);
    \draw[bracket] ($(#1,#2)+(0,\H)$) -- ($(#1,#2)+(0.28,\H)$);
    \draw[bracket] ($(#1,#2)+(0,0)$) -- ($(#1,#2)+(0.28,0)$);
    \draw[bracket] ($(#1,#2)+(\W,\H)$) -- ($(#1,#2)+(\W,0)$);
    \draw[bracket] ($(#1,#2)+(\W-0.28,\H)$) -- ($(#1,#2)+(\W,\H)$);
    \draw[bracket] ($(#1,#2)+(\W-0.28,0)$) -- ($(#1,#2)+(\W,0)$);
}

\begin{scope}[shift={(0,0)}]

    \drawmatrixbrackets{0}{0}

    \node[hblock, minimum width=1.15cm, minimum height=0.85cm] at (\xA,\yA) {};
    \node[charge] at (\xA,\yA) {$Q=0$};
    \node[charge] at (\xB,\yA) {$\cdots$};
    \node[charge] at (\xC,\yA) {$\cdots$};
    \node[charge] at (\xD,\yA) {$\cdots$};
    \node[charge] at (\xE,\yA) {$Q=(0,n)$};

    \node[charge] at (\xA,\yB) {$\vdots$};
    \node[charge] at (\xB,\yB) {$\ddots$};
    \node[charge] at (\xC,\yB) {$\cdots$};
    \node[charge] at (\xD,\yB) {$\cdots$};
    \node[charge] at (\xE,\yB) {$\vdots$};

    \node[charge] at (\xA,\yC) {$\vdots$};
    \node[charge] at (\xB,\yC) {$\cdots$};
    \node[charge] at (\xC,\yC) {$Q=i$};
    \node[charge] at (\xD,\yC) {$\cdots$};
    \node[charge] at (\xE,\yC) {$\vdots$};

    \node[charge] at (\xA,\yD) {$\vdots$};
    \node[charge] at (\xB,\yD) {$\cdots$};
    \node[charge] at (\xC,\yD) {$\cdots$};
    \node[charge] at (\xD,\yD) {$\ddots$};
    \node[charge] at (\xE,\yD) {$\vdots$};
    
    \node[charge] at (\xA,\yE) {$Q=(n,0)$};    
    \node[charge] at (\xB,\yE) {$\cdots$};
    \node[charge] at (\xC,\yE) {$\cdots$};
    \node[charge] at (\xD,\yE) {$\cdots$};
    \node[charge] at (\xE,\yE) {$Q=n$};

\end{scope}

\begin{scope}[shift={(\W+\gap,0)}]

    \drawmatrixbrackets{0}{0}

    \node[hblock, minimum width=1.15cm, minimum height=0.85cm] at (\xA,\yA) {};
    \node[charge] at (\xA,\yA) {$Q=0$};

    \node[hblock, minimum width=1.15cm, minimum height=0.85cm] at (\xB,\yB) {};

    \node[hblock, minimum width=1.15cm, minimum height=0.85cm] at (\xC,\yC) {};

    \node[hblock, minimum width=1.15cm, minimum height=0.85cm] at (\xD,\yD) {};
    \node[hblock, minimum width=1.15cm, minimum height=0.85cm] at (\xE,\yE) {};

    \node[charge] at (\xA,\yA) {$Q=0$};
    \node[charge] at (\xB,\yA) {$\cdots$};
    \node[charge] at (\xC,\yA) {$\cdots$};
    \node[charge] at (\xD,\yA) {$\cdots$};
    \node[charge] at (\xE,\yA) {$Q=(0,n)$};

    \node[charge] at (\xA,\yB) {$\vdots$};
    \node[charge] at (\xB,\yB+0.1) {$\ddots$};
    \node[charge] at (\xC,\yB) {$\cdots$};
    \node[charge] at (\xD,\yB) {$\cdots$};
    \node[charge] at (\xE,\yB) {$\vdots$};

    \node[charge] at (\xA,\yC) {$\vdots$};
    \node[charge] at (\xB,\yC) {$\cdots$};
    \node[charge] at (\xC,\yC) {$Q=i$};
    \node[charge] at (\xD,\yC) {$\cdots$};
    \node[charge] at (\xE,\yC) {$\vdots$};

    \node[charge] at (\xA,\yD) {$\vdots$};
    \node[charge] at (\xB,\yD) {$\cdots$};
    \node[charge] at (\xC,\yD) {$\cdots$};
    \node[charge] at (\xD,\yD+0.1) {$\ddots$};
    \node[charge] at (\xE,\yD) {$\vdots$};
    
    \node[charge] at (\xA,\yE) {$Q=(n,0)$};    
    \node[charge] at (\xB,\yE) {$\cdots$};
    \node[charge] at (\xC,\yE) {$\cdots$};
    \node[charge] at (\xD,\yE) {$\cdots$};
    \node[charge] at (\xE,\yE) {$Q=n$};
\end{scope}

\end{tikzpicture}
\caption{Schematic representation of the strong twirling (left) and weak twirling (right), in the charge basis associated with Eve's description. The black matrix elements denote the original kinematical density operator. The red blocks indicate the components retained by the projection onto the trivial-charge sector in the strong case, and onto the weakly invariant algebra in the weak case.}
\revtexfloatlabel{fig:symmetry}
\end{figure}

\subsection{Perspectives and frame-dependence of quantum channels}
\label{sec:perspectives-quantum-channels}
A final point we want to emphasise in this section is the role of perspective in the QRFs framework. A choice of perspective, or equivalently a choice of quantum reference frame, induces a preferred tensor-product structure of the relevant Hilbert space. In particular, given a QRF $i$ and the system $S\mid i$, the perspective of $i$ singles out a factorisation of the Hilbert space. For the PN approach\footnote{In the QI approach, the notion of perspective is not, in general, associated with a factorisation of the form \eqref{eq:factorisation}. Rather, it is implemented by the refactorisation map $\mathcal V_i^{(S|i)}$, which induces a preferred decomposition of the Hilbert space and of the weakly invariant algebra. Hence, the discussion of this section applies to the QI framework after replacing the simple PN factorisation by the corresponding refactorised algebraic structure.} the desired factorisation is of the form
\begin{equation} \label{eq:factorisation}
    \hilbert^{\mathrm{phys}} \cong \hilbert_{S\mid i} = \bigotimes_{j\neq i} \hilbert_j.
\end{equation}
Different choices of $i$ therefore identify $\hilbert^{\mathrm{phys}}$ with different tensor-product structures, i.e. subsystem relativity\nbs\cite{hoehn2021quantum,Hoehn_2023_subsystems,delahamette2021perspectiveneutral}. Since quantum channels are defined relative to a specified input-output tensor-product structure, changing perspective requires a consistent transformation of the channel. Therefore, quantum channels are frame-dependent objects. This was already apparent in Sec.\nbs\ref{sec:QI-approach}: after the refactorisation induced by $\mathcal{V}_i^{(S|i)}$, the form of the twirling changes because the gauge action is entirely absorbed by the left factor $C_L$.

We now introduce a framework describing how quantum channels transform under a change of QRF. This will be used in the subsequent analysis of the Paradox of the Third Particle. Let $i$ and $j$ be two QRFs whose change of perspective before the operation is implemented by the unitary operator $T_{i\rightarrow j}:\hilbert_{S|i}\rightarrow\hilbert_{S|j}$.\footnote{Here we assume that the two frames are related by a unitary transformation. Therefore, this discussion applies to all the approaches considered above when restricting to internal frames.} Let\footnote{For infinite-dimensional Hilbert spaces, states are described by trace-class operators, and we therefore take $\mathcal{E} : \mathcal{B}_1(\mathcal{H}_{S|i})\to \mathcal{B}_1(\mathcal{H}_{S'|i})$. The adjoint map is understood with respect to the trace duality, i.e. $\Tr[\mathcal E(\rho)X]=\Tr[\rho\,\mathcal E^*(X)]$, and acts on bounded operators as $\mathcal{E}^* : \mathcal{B}(\mathcal{H}_{S'|i})\to \mathcal{B}(\mathcal{H}_{S|i})$.} $\mathcal{E}:\mathcal B_1(\hilbert_{S|i})\to\mathcal{B}_1(\hilbert_{S'|i})$ be the CPTP map describing a quantum operation in the perspective of $i$, with Kraus operators $\{K_l\}_{l \in I}$ where $I$ is some index set and $S'$ is the output system of the channel. After the operation, the relevant QRF transformation may differ from the initial one, because the channel can change the output Hilbert space or the tensor-product structure under consideration. For example, a partial trace removes degrees of freedom and therefore changes the Hilbert space on which the output state is represented. We denote by $T'_{i\rightarrow j}:\hilbert_{S'|i}\rightarrow\hilbert_{S'|j}$ the corresponding post-operation change of perspective.\footnote{In the setting considered in this paper we can always find a unitary implementing the change of perspective after the application of the channel. In general this need not hold: for instance, when $\mathcal{E}$ is the partial trace over particle $j$ performed in the perspective of $i$, the output no longer contains the frame $j$ one wishes to transform to, and thus there is no such post-operation change of perspective $T'$.}

Since $\mathcal{E}$ is defined in the perspective of $i$, its representation in the perspective of $j$ is obtained by first mapping the input state from $j$ to $i$, then applying $\mathcal E$, and finally mapping the output back from $i$ to $j$. The operation described in the perspective of $j$ is therefore the CPTP map $\TEpsilon:\mathcal{B}_1(\hilbert_{S|j})\to\mathcal{B}_1(\hilbert_{S'|j})$ defined by the consistency requirement
\begin{equation}\label{eq:consistencychannel}
    \TEpsilon\!\left[\rho_{S|j}\right] \;=\; T_{i\rightarrow j}' \, \mathcal{E}\!\left[T_{i\rightarrow j}^\dagger \rho_{S|j} T_{i\rightarrow j}\right] \, T_{i\rightarrow j}'^{\dagger}.
\end{equation}
Since $\mathcal{E}$ admits the Kraus decomposition $\{K_l\}_{l \in I}$, then $\TEpsilon$ admits the Kraus decomposition
\begin{equation}
    \tilde{K}_l \;=\; T_{i\rightarrow j}' \, K_l \, T_{i\rightarrow j}^\dagger, \quad l \in I  .
\end{equation}
The proof follows immediately by expanding Eq.\nbs\eqref{eq:consistencychannel} with the Kraus decomposition of $\mathcal{E}$. Using the Kraus Representation Theorem\nbs\cite{renes2022quantum}, it follows that $\tilde{\mathcal E}$ is a CPTP map. Conceptually, the transformed channel consists of three steps: \textit{(1)} map from the perspective of $j$ to that of $i$, \textit{(2)} apply the channel in the perspective of $i$, and \textit{(3)} return to the perspective of $j$. This QRF-dependence of quantum channels will play a crucial role in the analysis of the Paradox of the Third Particle.

\section{The Paradox of the Third Particle}\label{sec:paradox}
We now turn to the main subject of this work: the Paradox of the Third Particle. The Paradox, originally introduced in\nbs\cite{Angelo_2011}, can be summarised conceptually as follows: we start from a description of two particles relative to Eve's external frame, and we encode some information $\theta$ in the relational degrees of freedom between particles 1 and 2. Clearly, the phase $\theta$ is accessible from the perspective of particle 1. Now consider the same situation, but with a third, uncorrelated particle added to the description. Intuitively, if the third particle is uncorrelated, we expect particle 1 to have access to the phase $\theta$ independently of whether particle 3 is included in the description. However, it turns out that the inclusion of this third particle affects the accessibility of the phase $\theta$ from the QRF of particle 1. Therefore, the Paradox can be summarised as follows: \textbf{it appears that whether or not one includes a third, uncorrelated particle in the description affects whether particle 1 has access to relational information between particles 1 and 2.}

Formally, we can introduce the Paradox of the Third Particle as follows: we start from a state describing particles 1 and 2 from the perspective of Eve. Let $\hilbert_{12}$ be the kinematical Hilbert space, and let $a,b \in G$ and $\theta \in \mathbb{R}$. Consider the following kinematical state\footnote{For continuous $G$, the expressions below are understood in the distributional sense; for finite $G$ they hold verbatim.}
\begin{equation}\label{eq:state-eve-2-paradox}
\ket{\psi}_{12} = \frac{1}{\sqrt{2}} \left[\ket{a^{-1}}_1\ket{b}_2+ e^{i\theta}\ket{a}_1\ket{b^{-1}}_2\right],
\end{equation}
Ref.\nbs\cite{Angelo_2011} describes the next step as a transition from absolute to relative positions. In the PN language, this corresponds to first extracting the physical component of the kinematical state through the strong twirling and normalising it\footnote{In the PN framework, the kinematical Hilbert space is an auxiliary gauge-redundant space, whereas the physical Hilbert space is the state space of the constrained theory. Given a kinematical state $\rho$, the operator $\Pi\rho\Pi$ represents only its subnormalised physical component, and the corresponding physical state is obtained by normalising this component. Since the physical Hilbert space carries the gauge-invariant content of the theory, one may equivalently take the normalised physical state as the starting point of the analysis. However, in this paper, we want to discuss the procedure of tracing from Eve's perspective, and we therefore start from kinematical states.} to obtain $\ket{\psi}_{12}^{\rm phys}$. Then we map to the perspective of particle 1 by applying the Schr\"odinger reduction map,
\begin{equation}
\ket{\psi}_{2|1}^{(e)}
:=
\mathcal{R}_{2}^{(1)}(e)\ket{\psi}_{12}^{\rm phys}.
\end{equation}
A straightforward calculation yields
\begin{equation} \label{eq:paradoxonly2}
\ket{\psi}_{2|1}^{(e)} = \frac{1}{\sqrt{2}} \left[\ket{ab}_2+ e^{i\theta}\ket{a^{-1}b^{-1}}_2\right]
\end{equation}
and we observe that the phase $\theta$ is accessible from the perspective of particle 1. This shows that the information on $\theta$ is encoded in the relational degrees of freedom between particles 1 and 2, namely in the component of the kinematical state supported on the trivial-charge sector of subsystem 12. We now add a third particle, uncorrelated with particles 1 and 2, to the description. We denote the kinematical Hilbert space of the three particles by $\hilbert_{123}$, and we assume that particle 3 is prepared in the orientation state labelled by $c\in G$. That is, the kinematical state reads
\begin{equation} \label{eq:stateeveparadox}
    \ket{\Psi}_{123} = \ket{\psi}_{12} \otimes \ket{c}_3 
\end{equation}
Projecting onto the trivial-charge sector, normalising the resulting physical state, and applying the Schr\"odinger reduction map with respect to particle 1 yields
\begin{equation} \label{eq:statein1paradox}
    \ket{\Psi}_{23|1}^{(e)} =   \frac{1}{\sqrt{2}} \left[\ket{ab}_2\ket{ac}_3+ e^{i\theta}\ket{a^{-1}b^{-1}}_2\ket{a^{-1}c}_3\right].  
\end{equation}
We observe that, for generic choices of the group elements, particle 2 becomes entangled with particle 3. Tracing particle 3 then results in the mixed state\footnote{This holds when the two branches associated with particle 3 are orthogonal, i.e. for $ac\neq a^{-1}c$. If $ac=a^{-1}c$, the reduced state need not be diagonal in this form.}
\begin{equation}
    \Tr_3\left[\ket{\Psi}\bra{\Psi}_{23|1}^{(e)}\right]
    =
    \frac{1}{2}\left[\ket{ab}\bra{ab}_2+\ket{a^{-1}b^{-1}}\bra{a^{-1}b^{-1}}_2\right].
\end{equation}
Thus, after tracing out particle 3, the phase $\theta$ is no longer accessible from the reduced description of particle 2 relative to particle 1. At first sight, this is surprising, since particle 3 was added in an uncorrelated state and one would therefore expect it not to affect the accessibility of relational information between particles 1 and 2. In the following, we examine why this Paradox arises, discuss its implications, and assess whether it constitutes a genuine Paradox.

\subsection{Where does the Paradox emerge?}\label{sec:paradox_emergence}

We now discuss precisely how and why the Paradox of the Third Particle emerges. This discussion is structured around two key points:

\noindent\textit{(1) Change of Hilbert space decomposition.} As discussed in \cref{sec:perspectives-quantum-channels}, quantum channels are frame dependent. The partial trace with respect to which the information about $\theta$ is preserved is defined in Eve's perspective. Changing QRF corresponds to a change of the TPS of the Hilbert space. Hence, if we want to obtain in this new QRF the same results as those obtained via the partial trace in Eve's QRF, we must transform the partial trace accordingly.

Indeed, one can note the following natural but important point: the Paradox of the Third Particle arises even if one assumes perfect knowledge of the relations between the frames shared by all parties, that is when the particles have access to Eve. To see this, we adopt the Perspectival approach\nbs\cite{Giacomini2019,deLaHamette2020QRF}. Starting from Eve's perspective and Eq.\nbs\eqref{eq:stateeveparadox}, one can take the perspective of particle 1 using the unitary $U_{E\rightarrow 1}$ of Eq.\nbs\eqref{eq:unitary_perspectival}, applied with $i=E$ and $j=1$. Tracing out particle 3 in the perspective of 1 again results in a loss of coherence, and hence in a loss of access to the phase $\theta$, even when the internal-frame descriptions include relative-frame data relating them to Eve.

Our interpretation of this point is that the trace has not been properly \enquote{re-framed} relative to the new QRF perspective. Therefore, in this setting, the Paradox is only apparent and can be resolved straightforwardly by providing a covariant notion of trace under QRF transformations, which is captured by the following definition:\footnote{This definition is an application of the formalism of \cref{sec:perspectives-quantum-channels}.}
\begin{definition}[Perspective Trace]\label{def:perspective-trace}
Let $i$ and $j$ be two QRFs and let $k\in S|i\cap S|j$. Let
$\Tr_k:\mathcal{B}_1(\hilbert_{S|i})\rightarrow\mathcal{B}_1(\hilbert_{S\setminus k|i})$
be the partial trace over system $k$ defined in the QRF $i$. The corresponding notion of discarding system $k$ seen by $i$ in the perspective of $j$ is the Perspective Trace, defined as the CPTP map
$\operatorname{T}^{\,k\mid i}_{\mathrm{PT},j}:\mathcal{B}_1(\hilbert_{S|j})\rightarrow\mathcal{B}_1(\hilbert_{S\setminus k|j})$
given by
\begin{equation}
\operatorname{T}^{\,k\mid i}_{\mathrm{PT},j}[\rho_{S\mid j}]
:=
T'_{i\rightarrow j}\,
\Tr_k\!\left[T_{i\rightarrow j}^{\dagger}\, \rho_{S\mid j} \, T_{i\rightarrow j}\right]\,
T_{i\rightarrow j}'^{\dagger}.
\end{equation}
Equivalently, if $\{\operatorname{K}_l\}_{l\in I}$ are Kraus operators for $\Tr_k$, then $\operatorname{T}^{\,k\mid i}_{\mathrm{PT},j}$ has Kraus operators
\begin{equation}
        \tilde{\operatorname{K}}_l =  T'_{i\rightarrow j} \operatorname{K}_l T_{i\rightarrow j}^{\dagger},\quad l \in I.
\end{equation}
The operators $T_{i\rightarrow j}:\hilbert_{S|i}\rightarrow\hilbert_{S|j}$ and $T'_{i\rightarrow j}:\hilbert_{S\setminus k|i}\rightarrow\hilbert_{S\setminus k|j}$ are the unitaries implementing the QRF transformation from $i$ to $j$ before and after the partial trace, respectively.\footnote{By construction, the output of the Perspective Trace $\operatorname{T}^{\,k\mid i}_{\mathrm{PT},j}[\rho_{S\mid j}]$ is unitarily related to $\Tr_{k}[\rho_{S\mid i}]$. Note that this does not contradict subsystem relativity: we do not relate the two outputs of native partial traces $\Tr_{k}[\rho_{S\mid j}]$ and $\Tr_{k}[\rho_{S\mid i}]$, which are indeed not unitarily related in general~\cite{hoehn2021quantum,Hoehn_2023_subsystems,delahamette2021perspectiveneutral}.}
\end{definition}
In the notation of the Perspective Trace, the superscript $k|i$ records that the map discards system $k$ seen by $i$; while the subscript carries two pieces of bookkeeping: $\mathrm{PT}$ names the operation, and $j$ records the QRF in which the partial trace is transported covariantly. We can apply the Perspective Trace to the problem above by considering $\operatorname{T}^{\,3\mid E}_{\mathrm{PT},1}:\mathcal{B}_1(\hilbert_{E23|1})\rightarrow\mathcal{B}_1(\hilbert_{E2|1})$ and using the unitary transformations from the Perspectival approach in the definition of the Kraus operators of the Perspective Trace. A straightforward calculation shows that the apparent loss of coherence in this setting---where the same relative-frame data are retained in all frame descriptions---is then resolved. We therefore identify the first pathological step underlying the Third-Particle Paradox as the use of a partial trace that has not been transformed covariantly under the QRF transformation. Accounting for the QRF covariance of the partial trace is equivalent to accounting for subsystem relativity.\footnote{Indeed, the relevant algebra seen by $i$, that is $O_{S\setminus k|i} \otimes \mathbb I_k$, is not, in general, mapped under the QRF transformation to the algebra $O_{S\setminus k|j} \otimes \mathbb I_k$ (cf. subsystem relativity\nbs\cite{hoehn2021quantum,Hoehn_2023_subsystems,delahamette2021perspectiveneutral}). Hence, Definition\nbs\ref{def:perspective-trace} ensures that, from the perspective of $j$, we trace out subsystem $k$ as seen by $i$, which in general does not coincide with subsystem $k$ as seen by $j$.} However, this is not the end of the story, as we discuss below.

\noindent\textit{(2) The partial trace does not preserve the physical Hilbert space.} The resolution of point \textit{(1)} is applicable under the condition that all frames possess perfect knowledge of inter-frame relations. However, in the PN approach, the situation is more subtle. The application of the inverse reduction map from the QRF 1 does not take us back to Eve’s kinematical perspective, but rather to the physical Hilbert space. The fundamental issue is that the partial trace in the physical space is not, in general, equivalent to the kinematical partial trace. Our aim is to compare the traces performed before and after physicalisation. Starting from Eve's notion of trace, it is not a priori clear how the projection onto the physical space affects the discarding operation. In fact, one can check that applying the partial trace after the restriction to the physical sector results in an operator that is not supported on the physical space:\footnote{Although the physical Hilbert space does not inherit the TPS from the kinematical space, we can embed physical states into the kinematical space and apply the usual partial trace there. More precisely, let \(J:\mathcal H^{\rm phys}_{123}\to\mathcal H_{123}\) denote this kinematical embedding, with \(J^\dagger J=\mathbb{I}_{\mathcal H^{\rm phys}_{123}}\) and \(JJ^\dagger=\Pi_{123}\). Given a physical state \(\rho^{\rm phys}_{123}\), we regard \(J\rho^{\rm phys}_{123}J^\dagger\) as a kinematical operator supported on the globally invariant sector. Hence, Eq.\nbs\eqref{eq:out-of-phys} shall be formally understood as $\Tr_3\left[J\ket{\Psi}\bra{\Psi}_{123}^{\rm phys}J^\dagger\right]$.}
\begin{equation}   \label{eq:out-of-phys}\Tr_3\left[\ket{\Psi}\bra{\Psi}_{123}^{\rm phys}\right] \notin \mathcal{B}_1(\hilbert_{12}^{\rm phys}).
\end{equation}
This is not surprising: the projector onto the physical space of 123 is non-local with respect to the kinematical split between subsystem 12 and particle 3, and therefore creates correlations between them. Tracing out particle $3$ then results in a loss of coherence and therefore in a loss of information, see Fig.\nbs\ref{fig:trace_first_after}. In the PN approach we cannot simply transform the partial trace by unitary conjugation, since we must account for the fact that the physical space does not, in general, inherit the TPS of the kinematical Hilbert space.

\begin{figure}[h]
    \centering
    \includegraphics[width=400pt]{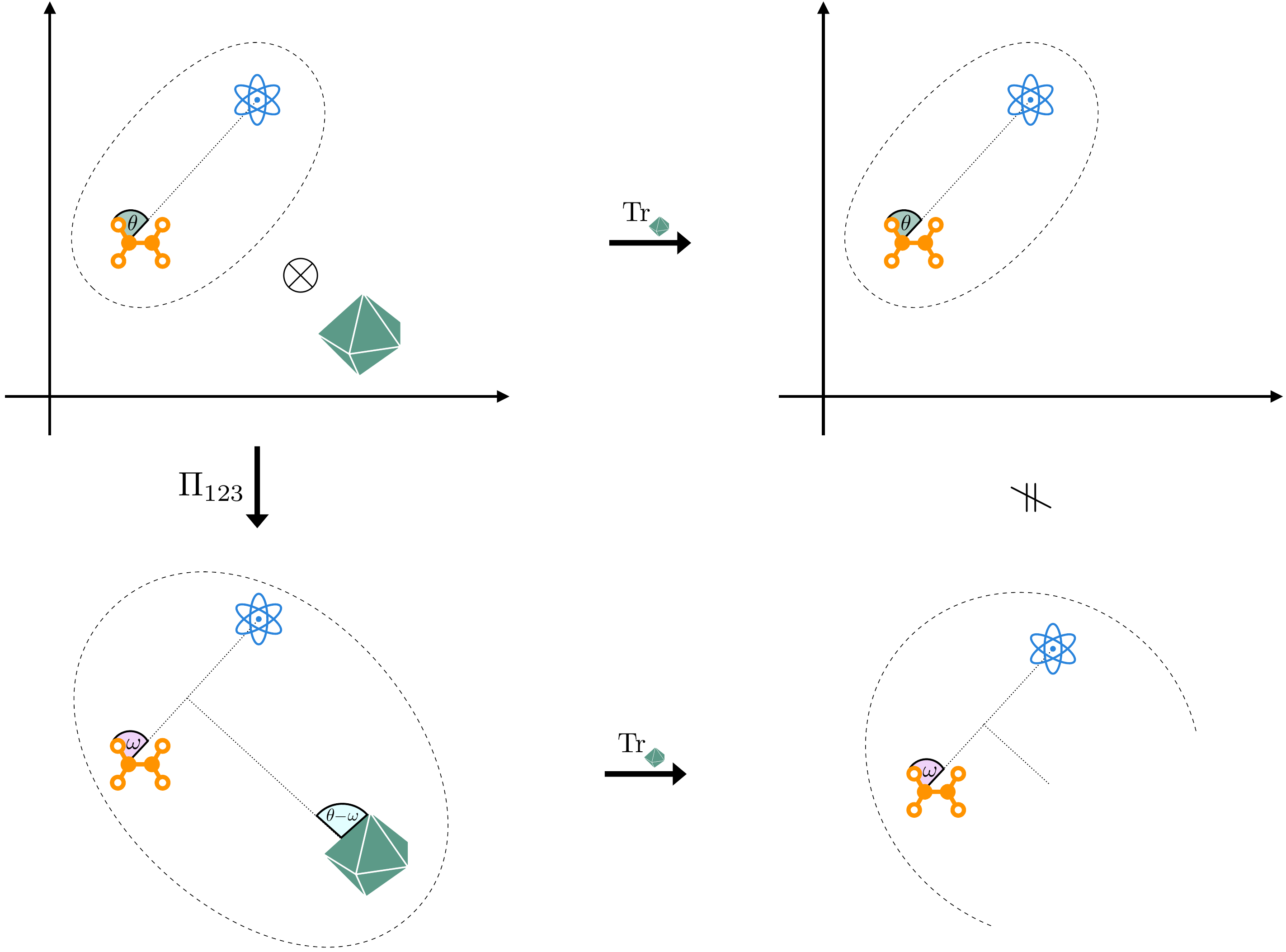}
    \caption{Heuristic explanation of the clash between tracing before and after projection onto the physical Hilbert space. The atom and the molecule represent subsystem 12, with the information $\theta$ encoded in the relative phase between them; the diamond represents the third particle. The upper pictures show Eve's perspective, where the Cartesian axes pictorially indicate an external reference frame; tracing from this external perspective recovers the information on $\theta$. In the bottom pictures the axes are absent, indicating that the global degrees of freedom have been discarded by the strong twirling. After strong twirling (bottom-left), the information on $\theta$ is spread onto correlations with the third particle, parametrised by a gauge-arbitrary phase $\omega$; tracing out particle 3 then results in the loss of $\theta$ (bottom-right).}
    \revtexfloatlabel{fig:trace_first_after}
\end{figure}

At first sight, this second problematic facet can be resolved via the \emph{Relational Trace}, introduced in Ref.~\cite{paradoxmuller}. The key role of this map is to reformulate the discarding operation at the level of physical---i.e.\ relational---degrees of freedom. Concretely, one considers states in $\mathcal{B}_1(\hilbert_{123})$ as described relative to Eve's external frame and observables $\mathcal{B}(\hilbert_{12})$ accessible to Eve. The key ingredient is an extended\footnote{With extended we mean that the embedding may not satisfy injectivity.} embedding map
$\Theta:\mathcal{B}(\hilbert_{12})\rightarrow\mathcal{B}(\hilbert_{123}^{\rm phys})$ defined as $\Theta(\mathcal{A}_{12}) = \Pi_{12} \mathcal{A}_{12} \Pi_{12} \otimes \Pi_3$, 
understood as mapping observables on $12$ into strongly-invariant observables supported on the physical Hilbert space of the enlarged system $123$. We now provide the definition of the \emph{Relational Trace}~\cite[Definition\nbs32]{paradoxmuller} generalised to compact Lie groups. 
\begin{definition}[Relational Trace]\label{def:relational_trace}
Let $\rho_{123}\in\mathcal{B}_1(\hilbert_{123})$ and let $\mathcal{A}_{12} \in \mathcal{B}(\hilbert_{12})$. Furthermore, let $\Theta:\mathcal{B}(\hilbert_{12})\rightarrow\mathcal{B}(\hilbert_{123}^{\rm phys})$ be an extended embedding map. The \emph{Relational Trace} (RT) is then defined as the unique adjoint of this embedding with respect to the trace pairing, $\operatorname{T}_{\rm rel(3)}:=\Theta^\dagger:\mathcal{B}_1(\hilbert_{123})\rightarrow\mathcal{B}_1(\hilbert_{12}^{\rm phys})$, so that the statistics-consistency condition 
\begin{equation}\label{eq:condition_reltrace}
\Tr\!\left[\left(\Theta\!\left[\mathcal{A}_{12}\right]\right)\rho_{123}\right]
\;=\;
\Tr\!\left[\mathcal{A}_{12}\;\operatorname{T}_{\rm rel(3)}\!\left[\rho_{123}\right]\right]
\end{equation}
holds for all $\rho_{123}\in\mathcal{B}_1(\hilbert_{123})$ and $\mathcal{A}_{12} \in \mathcal{B}(\hilbert_{12})$. 
\end{definition}
Then, we provide the Theorem characterising the Relational Trace, which is a generalisation to compact groups of\nbs\cite[Theorem\nbs33]{paradoxmuller}. The proof is analogous to that in Ref.\nbs\cite{paradoxmuller}: manipulate Eq.\nbs\eqref{eq:condition_reltrace} and then use the non-degeneracy of the trace pairing. 

\begin{theorem}[Characterisation of the RT]\label{thm:RT}
    The RT assumes the explicit form 
    \begin{equation}\label{eq:reltrace}
    \operatorname{T}_{\rm rel(3)}\!\left[\rho_{123}\right]
    \;=\;
    \Tr_3\!\left[(\Pi_{12}\otimes\Pi_{3})\,\rho_{123}\,(\Pi_{12}\otimes\Pi_{3})\right]
    \;=\;
    \Pi_{12}\,
    \Tr_3\!\left[\Pi_{3}\,\rho_{123}\,\Pi_{3}\right]\,
    \Pi_{12}.
    \end{equation}
\end{theorem}

It is crucial for what follows to note that, by the construction in Definition\nbs\ref{def:relational_trace}, the RT satisfies condition~\eqref{eq:condition_reltrace} for all kinematical states.

By Eq.\nbs\eqref{eq:reltrace} it follows that the image of the RT lies in $\mathcal{B}_1(\hilbert_{12}^{\rm phys})$.\footnote{Note that $\left(\Pi_{12}\otimes\Pi_{3}\right) \Pi_{123} = \Pi_{12}\otimes\Pi_{3}$} Indeed, the RT trace can be regarded as a disentangler between subsystem $12$ and $3$, followed by the standard partial trace. Furthermore, one can check that
\begin{equation}
    \operatorname{T}_{\rm rel(3)}\left[\ket{\Psi}\bra{\Psi}_{123}^{\rm phys}\right] \sim \ket{\psi}\bra{\psi}_{12}^{\rm phys} \in \mathcal{B}_1(\hilbert_{12}^{\rm phys}).
\end{equation}
After normalisation, applying the Relational Trace to the physicalisation of state~\eqref{eq:stateeveparadox} yields the same state as first tracing in Eve's perspective and only afterwards projecting onto the physical sector. From this standpoint, the second pathological step is resolved with this relational notion of subsystem discarding.

Given the presence of the tools to address points (1) and (2),\footnote{One might be tempted to regard the two problematic steps as two different rephrasings of the same issue, since point \textit{(2)} also modifies the Hilbert-space structure through the strong twirling. However, the two are conceptually different. Point \textit{(1)} concerns a change of perspective and is implemented by an isomorphism between Hilbert spaces, which correspondingly maps one TPS to another. Point \textit{(2)}, instead, concerns the projection onto the physical subspace, which is not an isomorphism and the physical subspace in general does not inherit a TPS.} it is reasonable to hypothesise that the Paradox of the Third Particle can be resolved for all states. However, as will be demonstrated in the following subsection, there are still crucial problems to be addressed. The aforementioned issues are rooted in the adoption of the Relational Trace as the concept of system discarding in order to address point \textit{(2)}. A more thorough examination of this topic is provided in the subsequent subsection.

\subsection{Limitation of the Relational Trace}\label{sec:rt_limitation}

The discussion of Sec.\nbs\ref{sec:paradox_emergence} offers a potential resolution to the Paradox. However, as we show in this section, this resolution is not satisfactory for all states. Since the RT satisfies Eq.\nbs\eqref{eq:condition_reltrace} for all kinematical states, one might expect it to resolve the Paradox for all states. We now provide a counterexample. Let $c,d\in G$ such that $c\neq d$ and define 
$\ket{-}_3:=\frac{1}{\sqrt{2}}\bigl(\ket{c}_3-\ket{d}_3\bigr)$.
Consider the state
\begin{equation}\label{eq:counterexample}
\ket{\tilde\Psi}_{123} = \ket{\psi}_{12}\otimes\ket{-}_3
\end{equation}
From Eve's kinematical perspective, subsystems $12$ and $3$ are uncorrelated. Therefore, we would expect this third uncorrelated particle to not affect whether particle\nbs1 has access to the phase $\theta$. Tracing out particle~3 from Eve's perspective leaves the reduced state of $12$ unchanged, and therefore the information about the phase $\theta$ is preserved from the perspective of particle\nbs1. On the other hand, for this choice one has $\Pi_{3}\ket{-}_3=0$, and therefore $\operatorname{T}_{\rm rel(3)}[\ket{\tilde\Psi}\bra{\tilde\Psi}_{123}]=0$.\footnote{The vanishing of
\(\Pi_3\ket{-}_3\) does not imply that the global physical projection
vanishes. For finite \(G\), and for generic choices of the group elements,
\(\ket{-}_3\) and \(\ket{\psi}_{12}\) have support in conjugate non-trivial
charge sectors, so that
\(\Pi_{123}\ket{\tilde\Psi}_{123}\neq0\).
For continuous compact \(G\), take a normalised
\(\ket{\psi_3}\in L^2(G)\) with vanishing Haar average,
\(\int_G dg\,\psi_3(g)=0\), so that
\(\Pi_3\ket{\psi_3}=0\), and let \(\bar q\) be a non-trivial irrep in which
it has a non-zero component. Choose a normalised
\(\ket{\psi_{12}}\in\mathcal H_{12}\) whose trivial-charge projection
carries the relational information \(\theta\), and whose \(q\)-sector
component has non-zero overlap, when coupled to the \(\bar q\)-sector
component of \(\ket{\psi_3}\), with the invariant vector
\(\ket{\omega_q}\); see App.~\ref{app:strong_twirling} and
App.~\ref{app:proof_proposition}. The global physical projection is
therefore non-zero in both cases, even though the RT, which restricts to
the factorised sector selected by \(\Pi_{12}\otimes\Pi_3\), annihilates the
state. In the terminology of Ref.~\cite[Theorem~34]{paradoxmuller}, the
corresponding conditional reduced operator is well defined but null.}
Thus the RT erases the relational information about subsystem $12$ that is preserved by first discarding particle~3 in Eve's kinematical description. Importantly, Eq.\nbs\eqref{eq:condition_reltrace} is nevertheless still satisfied: for the state~\eqref{eq:counterexample}, both sides vanish for all $\mathcal A_{12}\in\mathcal B(\hilbert_{12})$. The condition is therefore satisfied only trivially. We conclude that Eq.\nbs\eqref{eq:condition_reltrace} does not capture the operational content of the Third-Particle Paradox.

One might object that, in the PN approach, Eve's frame is pure gauge and therefore Eve's kinematical partial trace has no direct operational meaning. We agree with this statement. However, this is precisely the tension exposed by the Third-Particle Paradox. In both Refs.~\cite{Angelo_2011,paradoxmuller}, the third particle is first introduced in an external-frame description, which in the PN formulation adopted here corresponds to the kinematical level. The question is then whether the operation used to discard this additional system in the physical Hilbert space preserves the relational information about subsystem $12$ that would have been retained by discarding the third particle at the kinematical level. When we say that the RT does not resolve the Paradox for all states, we precisely intend that the RT does not preserve this relational information for some states, as shown above. We elaborate more on this in Sec.\nbs\ref{sec:PRT}.

In the following, we explain the origin of this failure. Specifically, we introduce a new statistical condition and a corresponding notion of trace. To capture the full operational content of the Paradox, we do not require the statistical condition to hold on the entire state space: we require that the relevant notion of perspectival trace preserves all subsystem relational information accessible to the external frame, without imposing any prior requirements on the states for which this condition can be solved. We then characterise the trace and discuss for which class of states our condition---and hence the Paradox---can be solved in the different approaches.

\section{Addressing the Paradox---Perspective Relational Trace}\label{sec:PRT}

We now present our proposal for a statistical consistency condition capturing the operational content of the Paradox, together with a corresponding notion of subsystem discarding that compares external and internal traces. We have seen that quantum operations, including the trace, transform covariantly under unitary QRF transformations, as made precise by the Perspective Trace. The change from an external to an internal perspective, by contrast, involves non-unitary operations, and it is therefore not a priori clear how to transform the partial trace consistently. In the following, we propose our resolution of the Paradox to address these two conceptually different points at the same time, both within the PN and QI frameworks. For notational clarity we specialise to the three-particle setting; all results hold verbatim for an arbitrary bipartition of $n$ particles into a subsystem and its complement. 
  
The statistical consistency condition is introduced, and the relevant notion of trace is defined as the map that satisfies it: the \emph{Perspective Relational Trace}. Specifically, we require that the expectation values assigned by the external reference frame to relational observables \(12\) coincide with those assigned in the perspective of particle~1 after the trace has been applied. A key assumption in the following is that subsystem \(12\) is described with the same framework as the total system \(123\). This will demonstrate that the PN approach does not, in general, provide a consistent description of a subsystem, whereas the QI approach can consistently reduce to any subsystem description. Notwithstanding the fact that the Paradox was originally formulated for product states, no restrictions are imposed on the state in the following.\footnote{This may appear stronger than the original formulation of the Paradox, which is motivated by product states and by the absence of correlation between \(12\) and \(3\). Here, however, we define the trace through a consistency requirement on a larger class of states, and then determine on which subclass of states this requirement can be satisfied.} 

\subsection{Strong twirling}\label{sec:paradox-strong-twirling}
We begin by discussing our proposed resolution within the PN approach. We introduce the Perspective Relational Trace, a linear map from the space describing $23$ to the one describing $2$, both from the perspective of particle~$1$.\footnote{To avoid notational clutter, from now on we suppress the explicit dependence on the conditioning orientation $g\in G$ in perspectival states, observables, reduction maps, and trace maps. All relations below shall be understood to hold for arbitrary orientation of the QRF.} We want to define a statistical consistency condition that precisely captures the operational content of the Paradox of the Third Particle. Recall that the relevant relational information $12$ is encoded in the trivial-charge sector, identified by observables $\mathcal{A}_{12}^{\text{phys}}$ in the algebra $\mathcal{B}(\hilbert^{\mathrm{phys}}_{12})$. Our goal is to find the subsystem-discarding operation in the perspective of particle~$1$ that preserves this subsystem relational information. Accordingly, we require that Eve’s measurement outcomes for $\mathcal{A}_{12}^{\mathrm{phys}}$ on states $\rho_{123} \in \mathcal{B}_1(\hilbert_{123})$ coincide with the statistics obtained in the perspective of particle~$1$, where observables $\mathcal{A}_{2|1}\in\mathcal{B}(\hilbert_{2|1})$\footnote{For ideal complete frames, $\mathcal{B}(\hilbert_{2|1})\cong \mathcal{B}(\hilbert_{2})$\nbs\cite{delahamette2021perspectiveneutral}.} are evaluated on the output state of the Perspective Relational Trace.  We formalise this intuition via the following definition. 

\begin{definition}[Perspective Relational Trace]\label{def:PRT_S}
The \emph{Perspective Relational Trace} (PRT) is a linear map\footnote{Since the discarded system and the perspective are fixed throughout, we simply write $\operatorname{T}^{3|1}_{\rm PRT}$ for the map discarding particle~3 in the perspective of particle~1, without the frame bookkeeping of Def.~\ref{def:perspective-trace}.}
\begin{equation}
\operatorname{T}_{\rm PRT}^{3|1} : \mathcal B_1(\mathcal H_{23|1}) \longrightarrow \mathcal B_1(\mathcal H_{2|1}).
\end{equation}
We say that $\operatorname{T}_{\rm PRT}^{3|1}$ \emph{satisfies the statistical consistency condition} for a kinematical state $\rho_{123} \in \mathcal B_1(\mathcal H_{123})$, with induced perspectival state $
\rho_{23|1} \;:=\; \mathcal R_{23}^{(1)}\,\Phi^{(123)}[\rho_{123}]\,\mathcal R_{23}^{(1)\dagger}$,
if, for every observable $\mathcal A_{2|1} \in \mathcal{B}(\hilbert_{2|1})$ with corresponding strongly-invariant observable $\mathcal A_{12}^{\mathrm{phys}} \in \mathcal{B}(\mathcal{H}_{12}^{\mathrm{phys}})$, it holds that
\begin{equation} \label{eq:statisticpn}
\Tr\!\left[(\mathcal A_{12}^{\mathrm{phys}}\otimes \mathbb I_3)\,\rho_{123}\right]
=
\Tr\!\left[\mathcal A_{2|1}\,\operatorname{T}_{\rm PRT}^{3|1}[\rho_{23|1}]\right].
\end{equation}
\end{definition}

The following Lemma, proven in App.\nbs\ref{app:proof_lemma_1}, shows that Eq.\nbs\eqref{eq:statisticpn} is equivalent to a state equivalence condition.
\begin{lemma} \label{lem:lemma_1}
The statistical consistency condition in Def.\nbs\ref{def:PRT_S} is equivalent to
\begin{align} \label{eq:lemma_1}
\mathcal R_{2}^{(1)} \Phi^{(12)} \left[\Tr_{3}\left[\rho_{123}\right]\right] \mathcal R_{2}^{(1)\dagger} &= 
\operatorname{T}_{\rm PRT}^{3|1} [\mathcal R_{23}^{(1)} \Phi^{(123)} \left[\rho_{123}\right] \mathcal R_{23}^{(1)\dagger}].
\end{align}
\end{lemma}
Furthermore, this lemma allows for a graphical representation of the statistical consistency condition, shown in Fig.\nbs\ref{fig:commuting_neutral}. We now characterise the Perspective Relational Trace and its uniqueness via the following Theorem, whose proof is given in App.\nbs\ref{app:proof_theorem_1}.

\begin{theorem}[Characterisation of the PRT]\label{thm:PRT_PN}
Let
\begin{equation}
\Lambda \;=\;
\Bigl\{\, \rho_{123} \in \mathcal{B}_1(\hilbert_{123}) \;\bigm|\; \Phi^{(12)}\!\left[\Tr_3\!\left[\Pi_{3}\,\rho_{123}\right]\right] = \Phi^{(12)}\!\left[\Tr_3\!\left[\rho_{123}\right]\right]
\Bigr\}
\end{equation}
There exists a unique linear map $\operatorname{T}_{\rm PRT}^{3|1}$ satisfying the statistical consistency condition~\eqref{eq:statisticpn} for every $\rho_{123} \in \Lambda$. This map is CP and admits the Kraus decomposition\footnote{Here $\{\operatorname K_l\}$ is a Kraus family of the partial trace $\Tr_3$ over $\hilbert_3=L^2(G)$, namely $\operatorname K_l=\mathbb I_{12}\otimes\langle l|_3$. For continuous $G$ the label $l\in G$ is to be read in a distributional sense: the kets $\{|g\rangle_3\}_{g\in G}$ are not normalisable elements of $L^2(G)$, so $\{\operatorname K_g\}_{g\in G}$ is a continuous (rigged) family rather than a countable Kraus set. A countable family is obtained in the Peter--Weyl basis $\{|q,m,n\rangle_3\}$ of $\hilbert_3$, indexed by $\alpha=(q,m,n)$, with $\operatorname K_\alpha=\mathbb I_{12}\otimes\langle q,m,n|_3$; the two are related by the Peter--Weyl expansion of $|g\rangle_3$ given in Sec.~\ref{sec:qrf_frameworks}.}
\begin{equation} \label{eq:Kraus_PRT_PN}
\tilde{\operatorname{K}}_l = \mathcal R_{2}^{(1)}\, \Pi_{12}\,\operatorname K_l\, \mathcal R_{23}^{(1)\dagger}, \quad l \in G,
\end{equation}
where $\{\operatorname K_l\}_{l\in G}$ are the Kraus operators of the partial trace $\Tr_3$ defined in Eve's perspective. Moreover, $\Lambda$ is the largest kinematical subspace containing the physical subspace on which Eq.\nbs\eqref{eq:statisticpn} can be satisfied.
\end{theorem}

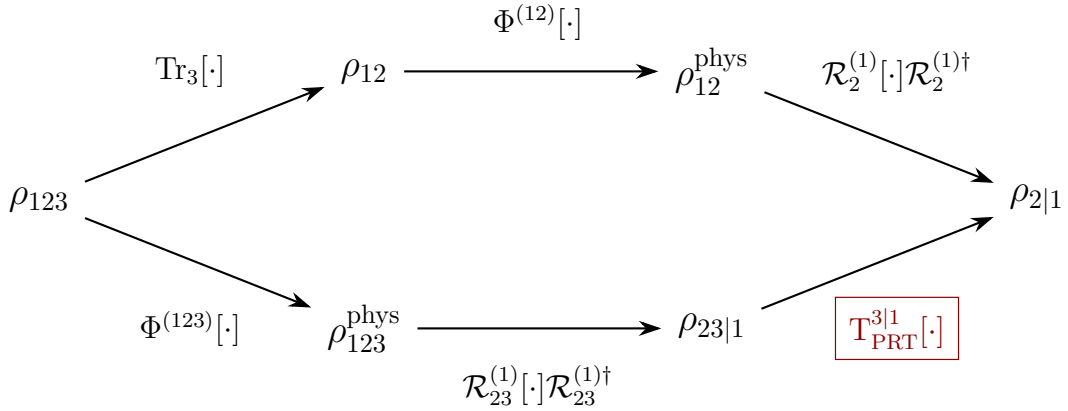
\begin{figure}[ht]
  \centering
    \begin{tikzpicture}[
        state/.style={
            font=\Large,
            inner sep=2pt
        },
        map/.style={
            font=\large,
            inner sep=1.5pt,
            fill=white
        },
        prtmap/.style={
            font=\large,
            text=red!55!black,
            draw=red!55!black,
            line width=0.45pt,
            fill=red!6,
            inner xsep=5pt,
            inner ysep=4pt,
            fill=white
        },
        arr/.style={
            -{Stealth[length=3.0mm,width=2.0mm]},
            line width=0.8pt,
            draw=black,
            shorten >=4pt,
            shorten <=4pt
        }
    ]

    \node[state] (rho123)     at (-6.65,  0.00) {$\rho_{123}$};

    \node[state] (rho12)      at (-2.35,  1.70) {$\rho_{12}$};
    \node[state] (rho12phys)  at ( 2.25,  1.70) {$\rho_{12}^{\rm phys}$};

    \node[state] (rho123phys) at (-2.35, -1.70) {$\rho_{123}^{\rm phys}$};
    \node[state] (rho23rel)   at ( 2.25, -1.70) {$\rho_{23|1}$};

    \node[state] (rho2rel)    at ( 6.55,  0.00) {$\rho_{2|1}$};

    \draw[arr] (rho123) -- (rho12);
    \draw[arr] (rho12) -- (rho12phys);
    \draw[arr] (rho12phys) -- (rho2rel);

    \draw[arr] (rho123) -- (rho123phys);
    \draw[arr] (rho123phys) -- (rho23rel);
    \draw[arr] (rho23rel) -- (rho2rel);

    \node[map] at (-4.65,  1.7)
        {$\Tr_{3}[\cdot]$};

    \node[map] at (-0.05,  2.35)
        {$\Phi^{(12)}[\cdot]$};

    \node[map] at ( 4.7,  1.7)
        {$\mathcal R_{2}^{(1)}[\cdot]\mathcal R_{2}^{(1)\dagger}$};

    \node[map] at (-4.65, -1.7)
        {$\Phi^{(123)}[\cdot]$};

    \node[map] at (-0.05, -2.45)
        {$\mathcal R_{23}^{(1)}[\cdot]\mathcal R_{23}^{(1)\dagger}$};

    \node[prtmap] at (4.7,-1.7)
        {$\operatorname{T}_{\rm PRT}^{3|1}[\cdot]$};

    \end{tikzpicture}
  \caption{Graphical interpretation of Eq.~(\ref{eq:lemma_1}). The statistical consistency condition is equivalent to requiring that the same state be obtained in the following cases: (Top arrows) trace, project into $\hilbert_{12}^{\rm phys}$ and map to the QRF of 1; (Bottom arrows) project into $\hilbert_{123}^{\rm phys}$, map to the QRF of 1 and apply the \emph{Perspective Relational Trace}, highlighted by the red box.}
\revtexfloatlabel{fig:commuting_neutral}
\end{figure}
The action of the Perspective Relational Trace can be written as 
\begin{equation}
        \operatorname{T}_{\rm PRT}^{3|1} \left[ \rho_{23|1} \right] =  \mathcal R_{2}^{(1)} \Phi^{(12)}\left[\Tr_3  \left[\mathcal R_{23}^{(1)\dagger} \rho_{23|1} \mathcal R_{23}^{(1)}\right]\right] \mathcal R_{2}^{(1)\dagger}.
\end{equation}
This highlights the sequence of operations underlying the map: \textit{(i)} map from the perspective of 1 to the physical space, \textit{(ii)} trace out subsystem $3$, \textit{(iii)} map to the physical space, and \textit{(iv)} map back to the perspective of 1. 

We now note the following technical, but important, point. The statistical 
condition in Def.~\ref{def:PRT_S} should be understood as a strong, 
unnormalised statistical consistency condition. Indeed, Eq.~\eqref{eq:statisticpn} does not merely compare the normalised relational state of subsystem $12$; it also requires agreement of the physical weights carried by the corresponding subnormalised representatives.\footnote{More explicitly, Eq.~\eqref{eq:statisticpn} compares the subnormalised physical representatives $X$ and $Y$ themselves, not only their normalised versions $X/\Tr[X]$ and $Y/\Tr[Y]$. For instance, if $X=p\tau$ and $Y=q\tau$, with $\tau$ a normalised physical state and $p,q>0$, then $X/\Tr[X]=Y/\Tr[Y]=\tau$. Nevertheless, Eq.~\eqref{eq:statisticpn} requires $X=Y$, and hence $p=q$, as is seen by choosing $\mathcal A_{12}^{\rm phys}=\Pi_{12}$.}

If, however, kinematical states are interpreted as redundant representatives 
of physical states, then the probabilistic content is carried by the normalised 
physical projection. Hence two kinematical states should be identified whenever 
their normalised physical projections coincide, or equivalently whenever their non-zero physical projections differ only by an overall positive 
scalar. This motivates the 
following normalised version of the PRT output:
\begin{equation}
\widehat{\operatorname T}_{\rm PRT}^{3|1}[\rho_{23|1}]
:=
\frac{
\operatorname T_{\rm PRT}^{3|1}[\rho_{23|1}]
}{
\Tr\!\left[\operatorname T_{\rm PRT}^{3|1}[\rho_{23|1}]\right]
},
\end{equation}
whenever the denominator is non-zero, and we extend the definition by setting $\widehat{\operatorname T}_{\rm PRT}^{3|1}[\rho_{23|1}]:= 0$ whenever $\Tr\!\left[\operatorname T_{\rm PRT}^{3|1}[\rho_{23|1}]\right] = 0$. Equivalently, using the explicit form of 
$\operatorname T_{\rm PRT}^{3|1}$, this can be written as
\begin{equation}
   \widehat{\operatorname T}_{\rm PRT}^{3|1}[\rho_{23|1}]
   =
   \mathcal R_{2}^{(1)}
   \widehat{\Phi}^{(12)}
   \left[
      \Tr_3
      \left[
        \mathcal R_{23}^{(1)\dagger}
        \rho_{23|1}
        \mathcal R_{23}^{(1)}
      \right]
   \right]
   \mathcal R_{2}^{(1)\dagger}.
\end{equation}
Note that $\mathrm T^{3|1}_{\rm PRT}$ is the linear CP trace-non-increasing map characterised above, whereas $\widehat{\mathrm T}^{3|1}_{\rm PRT}$ is trace-preserving whenever its denominator is non-zero, but is non-linear and not CP.

This technical but crucial observation leads to a weaker, normalised statistical condition for the 
Paradox of the Third Particle:
\begin{equation} \label{eq:normalised_condition}
\Tr\!\left[
\mathcal A_{12}^{\mathrm{phys}}\,
\widehat{\Phi}^{(12)}\!\left[\Tr_3[\rho_{123}]\right]
\right]
=
\Tr\!\left[
\mathcal A_{2|1}\,
\widehat{\operatorname{T}}_{\rm PRT}^{3|1}[\rho_{23|1}]
\right],
\end{equation}
where $\rho_{23|1}$ is the perspectival state obtained from the physical 
projection of $\rho_{123}$.\footnote{Strictly speaking, for this condition to have a natural probability interpretation $\rho_{23|1}$ should be interpreted as the normalised state 
of particles $23$ in the perspective of particle $1$. However, one can equivalently work with the unnormalised version since the normalised PRT is insensitive to the normalisation 
of its input.} Eq.\nbs\eqref{eq:normalised_condition} can be seen as a relaxation of Def.\nbs\ref{def:PRT_S} to allow expectation values to coincide up to an overall positive constant. Eq.\nbs\eqref{eq:normalised_condition} is satisfied if and only if the kinematical state $\rho_{123} \in \mathcal{B}_1(\hilbert_{123})$ is in the kinematical set\footnote{Here, $X\sim Y$ means that there exists $\tau>0$ such that $X=\tau Y$.}
\begin{equation}
\widehat{\Lambda} \;=\;
\Bigl\{ \; \rho_{123} \in \mathcal{B}_1(\hilbert_{123}) \;\bigm|\; \Phi^{(12)}\left[\Tr_3\left[\Pi_{3}\rho_{123}\right]\right] \sim  \Phi^{(12)}\left[\Tr_3\left[\rho_{123}\right]\right]
\Bigr\},
\end{equation}
Note that $\Lambda \subsetneq \widehat{\Lambda}$.\footnote{Indeed, $\Lambda \subseteq \widehat{\Lambda}$ trivially but $\ket{\Psi}\bra{\Psi}_{123}\in \widehat{\Lambda}$ while $\ket{\Psi}\bra{\Psi}_{123} \not \in \Lambda$. } Hence, this observation is not merely technical: it enlarges the set of states for which the Paradox can be resolved. One can check immediately that the Paradox is solved by applying the normalised PRT to the relational state $\ket{\Psi}_{23|1}$ of Eq.\nbs\eqref{eq:statein1paradox}.
This discussion tells us that, for kinematical states outside $\widehat \Lambda$, our statistical consistency condition cannot be satisfied. This gives an important result: the Paradox of the Third Particle cannot be resolved for all states in the PN approach. Although outside $\widehat \Lambda$ there are states with non-trivial correlations between subsystem $12$ and $3$ for which it is unsurprising that the Paradox cannot be resolved, this characterisation shows that there are product states for which the Paradox cannot be solved either. In particular, the product state~\eqref{eq:counterexample} that showed the failure of the RT does not lie in $\widehat \Lambda$. 

Since we were not able to solve the Paradox for Eq.~\eqref{eq:counterexample}, one might be tempted to conclude that the PRT does not provide further insights beyond what is already provided by the RT. Let us then spell out the key differences and clarify how the PRT sheds light on the problem. To make explicit the relation between the two notions of trace, it is useful to introduce the auxiliary map $\xi:\mathcal{B}_1(\hilbert_{123})\to\mathcal{B}_1(\hilbert_{2|1})$, defined as
\begin{align} 
    \xi\left[\cdot \right] & = \operatorname{T}_{\rm PRT}^{3|1} \left[\mathcal R_{23}^{(1)}\Phi^{(123)}\left[\cdot\right]\mathcal R_{23}^{(1)\dagger}\right] = \mathcal R_{2}^{(1)}\operatorname{T}_{\rm rel(3)}\left[\cdot\right] \mathcal R_{2}^{(1)\dagger}\label{eq:xi}
\end{align}
Therefore, at the level of the map itself, the PRT does not introduce a fundamentally new operation; rather, it is the RT dressed by the isometries to account for the QRF covariance of the partial trace. The crucial difference lies instead in the statistical condition used to characterise it. While the RT characterisation can become trivial on a class of states, our consistency condition makes explicit the mismatch between the external trace and the trace after physicalisation. We now discuss this difference.

Recall that the statistical condition---Eq.\nbs\eqref{eq:condition_reltrace}---characterising the RT holds for every state, while the PRT solves Eq.\nbs\eqref{eq:statisticpn} only for kinematical states in $\Lambda$.\footnote{For simplicity, here we omit the normalisation requirement in our consistency condition. This does not affect the argument of the following discussion.} This is not a limitation, but rather the \emph{key feature} of our approach. Ultimately, this reveals a limitation hidden by the RT condition. For clarity, let us compare the two statistical conditions side by side:
\begin{equation}\label{eq:comparison_conditions}
\begin{array}{c@{\qquad\qquad}c}
\text{PRT statistical consistency condition}
&
\text{RT statistical consistency condition}
\\[0.6em]
\Tr\!\left[\left(\mathcal A_{12}^{\mathrm{phys}}\otimes \mathbb I_3\right)\rho_{123}\right]
=
\Tr\!\left[\mathcal A_{2|1}\,\operatorname{T}_{\rm PRT}^{3|1}[\rho_{23|1}]\right]
&
\Tr\!\left[\left(\mathcal{A}_{12}^{\mathrm{phys}}\otimes \Pi_{3}\right)\rho_{123}\right]
\;=\;
\Tr\!\left[\mathcal{A}_{12}\;\operatorname{T}_{\rm rel(3)}\!\left[\rho_{123}\right]\right].
\end{array}
\end{equation}

One can check that, for the state\nbs\eqref{eq:counterexample}, the condition of Ref.~\cite{paradoxmuller} is trivial. The RT condition is evaluated entirely after restriction to the physical space, obscuring the comparison between tracing before and after physicalisation, which lies at the heart of the Paradox of the Third Particle, as discussed in Sec.\nbs\ref{sec:rt_limitation}. By contrast, for this choice of state, only the right-hand side is trivial in the PRT consistency condition. The left-hand side is not trivial, since particle $3$ is traced before projection onto the physical space: the clash between the two sides of the PRT condition shows the clash between the trace before and after physicalisation and ultimately that the Paradox cannot be resolved for this state. This highlights that the reduction implemented by the RT condition is not innocuous: it can annihilate a class of states for which tracing out particle $3$ in Eve's perspective yields a reduced state on subsystem $12$ that still carries relational information detectable from Eve's perspective.

Furthermore, the RT probes only the factorised sector selected by $\Pi_{12}\otimes\Pi_{3}$. Since this sector is, in general, a proper subspace of the full physical Hilbert space, the RT can annihilate states carrying relational information that is relevant in the full relational theory. Note that this statement is at a different level from the observation that the RT may erase relational information visible in Eve's reduced description after tracing out particle $3$.

More precisely, in the PRT consistency condition, the left-hand side probes the relational content of subsystem $12$ while leaving particle $3$ unrestricted, through the factor $\mathbb I_3$. In the RT condition, instead, particle $3$ is also projected onto its trivial-charge sector through $\Pi_3$. This difference is exactly what, in our condition, makes transparent the clash between tracing particle $3$ before physicalisation and after physicalisation. While the right-hand sides of the two conditions coincide everywhere,\footnote{The right-hand sides of the two conditions are both evaluated entirely on the physical sector.} the left-hand sides coincide only on $\Lambda$. To conclude this comparison, we stress that our statistical condition makes explicit two facts: \textit{(1)} the QRF-covariance of the partial trace\footnote{In Ref.~\cite{paradoxmuller}, this requirement is not formulated as a covariance condition for the trace, since the Relational Trace is constructed at the relational, perspective-neutral level. The relevant property is instead that the Relational Trace is \enquote{invariant} and preserves observational equivalence: it assigns the same reduced operator to observationally equivalent inputs~\cite[Theorem~33]{paradoxmuller}, and hence, in particular, to symmetry-equivalent ones~\cite[Definition~13]{paradoxmuller}. In their setting, the canonical representatives of an alignable state relative to two distinct internal frames are related by a symmetry $U \in U_{\mathrm{sym}}$, which implements the QRF transformation between the two frames \cite[Lemma~5, Definition~15, Theorem~18]{paradoxmuller}, and are therefore symmetry-equivalent in the sense of Ref.~\cite{paradoxmuller}. More specifically, their invariance condition can be seen as an instance of our covariance requirement, Eq.~\eqref{eq:consistencychannel}, in the special case where the \emph{post-operation} conjugation $T'_{A' \to B'}$ acts trivially on the output of the channel. This is the case in their setting because the output of the Relational Trace lies in the physical space, on which every $U \in U_{\mathrm{sym}}$ acts as the identity; the unitaries relating symmetry-equivalent inputs are elements of $U_{\mathrm{sym}}$, so no reframing of the output is needed and Eq.~\eqref{eq:consistencychannel} reduces to the invariance of the RT construction.} and \textit{(2)} the clash between tracing before and after physicalisation. The obstruction stems from the fact that strong twirling introduces correlations between subsystem $12$ and the otherwise uncorrelated third particle. As a result, after physicalisation, relational information about subsystem $12$ that was accessible from Eve's perspective may be encoded in correlations with particle $3$, and is therefore lost in the reduced description when particle $3$ is discarded (See Fig.\nbs\ref{fig:trace_first_after}).

We conclude this subsection by remarking that the impossibility of resolving the Paradox outside the set $\widehat \Lambda$ can be traced back to insisting on a PN description for both the composite system $123$ and subsystem $12$. We discuss this in Sec.\nbs\ref{sec:discussion}.

\subsection{Weak twirling}\label{sec:weak-twirling-PRT}

In this section, we extend the analysis of the previous section to the QI approach discussed in Sec.~\ref{sec:QI-approach}. Both the composite system $123$ and, after discarding particle $3$, the subsystem $12$ are now described in the QI framework. The statistical consistency condition is motivated as before: we require that Eve's measurement outcomes on the weakly-invariant algebra of subsystem $12$ coincide with those obtained in the perspective of particle\nbs1 on the perspectival algebra, after applying the relevant trace from the perspective of particle\nbs1. The crucial difference is that, in the QI approach, this information is encoded in the weakly-invariant algebra, rather than in the strongly-invariant algebra. We now define the corresponding QI version of the Perspective Relational Trace.

\begin{definition}[Perspective Relational Trace--QI]\label{def:PRT_QI}
The Perspective Relational Trace in the QI approach (PRT--QI) is a linear map
\begin{equation}
\operatorname{T}_{\rm PRT,QI}^{3|1}:\mathcal B_1(\mathcal H_{C,23|1})^G\longrightarrow\mathcal B_1(\mathcal H_{C,2|1})^G.
\end{equation}
We say that $\operatorname{T}_{\rm PRT,QI}^{3|1}$ satisfies the \emph{statistical consistency condition} for a kinematical state $\rho_{123}\in\mathcal B_1(\mathcal H_{123})$, with induced perspectival state $\rho_{C,23|1} =\mathcal{V}_1^{(23|1)}\!\left[\mathcal G^{(123)}[\rho_{123}]\right]$, if, for every observable $\mathcal A_{C,2|1}\in\mathcal B(\mathcal H_{C,2|1})^G$ with corresponding weakly invariant observable $\mathcal A_{12}^{\rm inv}\in\mathcal B(\mathcal H_{12})^G$, it holds that
\begin{equation} \label{eq:statistic_QI}
\Tr\!\left[\left(\mathcal A_{12}^{\rm inv}\otimes\mathbb I_3\right)\rho_{123}\right]
=
\Tr\!\left[\mathcal A_{C,2|1}\,\operatorname{T}_{\rm PRT,QI}^{3|1}\!\left[\rho_{C,23|1}\right]\right].
\end{equation}
\end{definition}

The following Lemma, in analogy with the PN case, shows that the statistical consistency condition can again be reduced to a state condition. The proof can be found in App.\nbs\ref{app:proof_lemma_2}.
\begin{lemma} \label{lem:lemma_2}
    The statistical consistency condition of Def.\nbs\ref{def:PRT_QI} is equivalent to
    \begin{align} \label{eq:lemma_2}
    \mathcal V_1^{(2|1)}\!\left[\mathcal G^{(12)}\!\left[\Tr_3[\rho_{123}]\right]\right]
    =
    \operatorname{T}_{\rm PRT,QI}^{3|1}\!\left[
    \mathcal V_1^{(23|1)}\!\left[\mathcal G^{(123)}[\rho_{123}]\right]
    \right].
    \end{align}
\end{lemma}
We characterise the Perspective Relational Trace QI and its uniqueness via the following theorem, proven in App.\nbs\ref{app:proof_theorem_2}.
\begin{theorem}[Characterisation of the PRT--QI]\label{thm:PRT_QI}
    There exists a unique linear map $\operatorname{T}^{3|1}_{\mathrm{PRT,QI}}$ satisfying the statistical consistency condition\nbs\eqref{eq:statistic_QI} for every $\rho_{123} \in \mathcal{B}_1(\hilbert_{123})$. This map is CPTP and admits the following Kraus decomposition\footnote{The Kraus operators shall again be understood in a distributional sense.}
    \begin{equation} \label{eq:Kraus_PRT_QI}    
    \tilde{\operatorname{K}}_l=V_1^{(2|1)} \operatorname{K}_l V_1^{(23|1)\dagger}, \; l \in G.
    \end{equation}
\end{theorem} 
The action of the Perspective Relational Trace QI can be written as 
\begin{equation} \label{eq:PRT_QI_action}
        \operatorname{T}_{\rm PRT,QI}^{3|1} \left[ \rho_{C,23|1} \right] 
        =
        \mathcal V_1^{(2|1)}
        \left[
        \Tr_3\left[
        \left(\mathcal V_1^{(23|1)}\right)^{-1}\left[\rho_{C,23|1}\right]
        \right]
        \right].
\end{equation}
This makes explicit the sequence of operations underlying the map: \textit{(1)} map from the perspective of particle $1$ to the weakly invariant algebra, \textit{(2)} trace out particle $3$, and \textit{(3)} map back to the perspective of particle $1$. Note that here there is no need to introduce a normalised statistical condition, since the weak twirling is trace-preserving. Eq.\nbs\eqref{eq:PRT_QI_action} shows that no twirling is required after the partial trace in the QI approach. The reason is that the partial trace of a weakly invariant trace-class operator is again weakly invariant. This is a general fact: for any weakly invariant $\sigma_{123}\in\mathcal{B}_1(\hilbert_{123})^G$ one has $\sigma_{123}=\mathcal{G}^{(123)}[\sigma_{123}]$, so that Eq.\nbs\eqref{eq:proof intertwine} gives $\Tr_3[\sigma_{123}]=\mathcal{G}^{(12)}[\Tr_3[\sigma_{123}]]\in\mathcal{B}_1(\hilbert_{12})^G$. 

Hence, the partial trace is well defined on weakly invariant trace-class operators. In this sense, $\operatorname{T}_{\rm PRT,QI}^{3|1}$ is simply the perspectival transformation of the ordinary partial trace. Crucially, in the QI approach the Paradox of the Third Particle can be resolved for \emph{all} kinematical states. The reason can be understood by looking at Eq.\nbs\eqref{eq:statistic_QI}: the observable $\mathcal A_{12}^{\rm inv}\otimes\mathbb{I}_3$ is weakly invariant on the full system $123$, so by Eq.\nbs\eqref{eq:HSproductQI} the left-hand side of Eq.\nbs\eqref{eq:statistic_QI} depends on $\rho_{123}$ only through $\sigma_{123}=\mathcal{G}^{(123)}[\rho_{123}]$, exactly as the right-hand side does. By contrast, in the PN condition\nbs\eqref{eq:statisticpn} the observable $\mathcal A_{12}^{\rm phys}\otimes\mathbb{I}_3$ is \emph{not} strongly invariant on $123$ --- the factor $\mathbb{I}_3$ is not the trivial-charge projector $\Pi_3$ --- so its left-hand side depends on the kinematical representative $\rho_{123}$, not merely on its physical projection. This mismatch between the two sides is precisely the origin of the restriction to $\Lambda$. We can argue therefore that in the QI approach there is no genuine Paradox of the Third Particle: the PRT-QI is just a perspective trace. 

Recall from App.\nbs\ref{app:qi_details} that, in the refactorised representation, the weakly invariant algebra decomposes as $\mathcal{B}(\hilbert_{C,2|1})^G\cong  \mathcal{O}_{2|1} \otimes \overline{O_{2|1}}$, where operators in $\mathcal{O}_{2|1}$ in the refactorised representation are of the form $T_{C,2|1} = \mathbb{I}_C \otimes T_{2|1}$ acting trivially on the extra-particle. Applying the map $\operatorname{T}_{\rm PRT,QI}^{3|1}$ to the weakly invariant state associated with $\ket{\Psi}\bra{\Psi}_{123}$ in the perspective of particle $1$, one obtains two contributions. The first contribution is $A:=\frac{1}{2\volg} \ \mathbb{I}_C \otimes \left[\ket{a  b}\bra{a  b}_2+\ket{a^{-1}  b^{-1}}\bra{a^{-1}  b^{-1}}_2\right]$, and clearly $A\in \mathcal{O}_{2|1}$. The second contribution is $B
:=
\frac{1}{2\volg}\int_G dg\,
e^{i\theta}
\ket{g a}\bra{g a^{-1}}_C
\otimes
\ket{a^{-1}b^{-1}}\bra{ab}_2
+\mathrm{h.c.}$: this term is non-trivial on the extra-particle factor, and therefore $B\notin \mathcal{O}_{2|1}$. Thus, the information on $\theta$ is not contained in the subalgebra of \enquote{relational observables} in the strict sense, acting only on the relational subsystem $2|1$; rather, it is contained in the full weakly invariant algebra through correlations with the extra-particle degree of freedom. Tracing out this additional degree of freedom removes the term $B$ for generic choices of the group elements, and hence removes the coherence carrying the phase $\theta$ from the reduced description. In this sense, the Paradox would appear in the QI approach if the extra-particle degree of freedom were discarded.

\subsection{Discussion and comparison between strong and weak twirling}\label{sec:discussion}

In the previous sections, we introduced the PRT as an operational notion of subsystem discarding across QRF perspectives. In the PN approach, we showed that the Third-Particle Paradox cannot be resolved, in its original operational formulation, for arbitrary states. We then argued that our statistical consistency condition captures the operational content of the Paradox---namely, the mismatch between external and internal notions of subsystem discarding---and compared it with the condition underlying the RT of Ref.~\cite{paradoxmuller}. By contrast, in the QI approach, the Paradox can be resolved for all states: the PRT--QI satisfies the corresponding statistical consistency condition for every input state.

Our analysis shows that insisting on describing subsystem $12$ within the strong-twirling framework, after part of a larger physical system has been discarded, leads in general to an inconsistency with subsystem reduction. More precisely, the PN approach provides the appropriate framework for describing a closed system as a whole: it imposes a global constraint on the total system and selects the trivial-charge sector as the physical one. Such a description is natural when the system of interest represents the \enquote{whole universe}, so that no feature of the \enquote{whole}---for instance, the charge---has operational significance. However, if one is interested in describing a proper subsystem of a larger physical system, the strong-twirling description is not, in general, stable under discarding. In this case, the natural reduced description is instead weakly invariant.

This obstruction can be illustrated through the following toy model. For illustration, we now specialise to the translation group, thereby dropping the assumption that \(G\) is compact. Consider a three-particle state: assuming that the three particles form an isolated system, we describe them within the PN approach and impose the global constraint
\begin{equation}
    P_{\rm tot}=p_1+p_2+p_3=0,
\end{equation}
so that physical states lie in the zero-total-momentum sector.  Suppose now that we are only interested in describing subsystem \(12\). What is the constraint on the momentum of this subsystem? We should clearly not require the subsystem momentum to vanish, since it is correlated with the momentum of particle \(3\): the global constraint only implies 
$p_1+p_2=-p_3$. In other words, the global strong-twirling projector clearly does not impose a vanishing charge on subsystem \(12\) alone. Rather, it correlates the charge sector of subsystem \(12\) with the compensating charge sector of particle \(3\). This is expressed by the following decomposition of the global projector onto the trivial-charge sector. For simplicity, we write it for a compact Abelian symmetry group; the general compact case is discussed in App.~\ref{app:strong_twirling}:
\begin{equation}\label{eq:strong_projector_charge_decomposition}
\Pi_{123}
=
\sum_{q\in \hat{G}}
\Pi_{12}^{(q)}\otimes \Pi_{3}^{(-q)} .
\end{equation}
Here, $\Pi_i^{(q)}$ projects onto the charge-$q$ sector of subsystem $i$. It follows immediately that tracing out \(3\) leaves subsystem \(12\) in a mixture of charge sectors and therefore, in general, in a state that is not supported on the physical subspace of subsystem \(12\). In particular, restricting to the factorised sector selected by \(\Pi_{12}\otimes \Pi_3\), as in the RT characterisation, would amount, in the translation example, to imposing $
p_1+p_2=0$ and $p_3=0$,
thereby erasing the contributions from the nontrivial sectors, i.e. those with \(p_1+p_2\neq 0\). How, then, should subsystem \(12\) be described? Let \(J:\mathcal H^{\rm phys}_{123}\to\mathcal H_{123}\) denote the canonical kinematical embedding, with \(J^\dagger J=\mathbb I_{\mathcal H^{\rm phys}_{123}}\) and \(JJ^\dagger=\Pi_{123}\). Given a physical state \(\rho^{\rm phys}_{123}\), we regard \(\tilde \rho_{123}^{\rm phys} := J\rho^{\rm phys}_{123}J^\dagger\) as a kinematical operator supported on the globally invariant sector. For the translation group, using the momentum representation, such a state can be written as
\begin{equation}
    \tilde \rho_{123}^{\rm phys}
    =
    \int dp_1dp_2dp'_1dp'_2 \,
    \varphi(p_1,p_2,p'_1,p'_2)
    \ket{p_1}\bra{p'_1}_1
    \otimes
    \ket{p_2}\bra{p'_2}_2
    \otimes
    \ket{-p_1-p_2}\bra{-p'_1-p'_2}_3 .
\end{equation}
Tracing out particle \(3\) and performing the change of variables \(P=p_1+p_2\), \(p_1=r\), and \(p'_1=r'\), we obtain
\begin{equation}
    \Tr_3\left[\tilde \rho_{123}^{\rm phys}\right]
    =
    \int dPdrdr' \,
    \vartheta(P,r,r')
    \ket{r}\bra{r'}_1
    \otimes
    \ket{P-r}\bra{P-r'}_2 ,
\end{equation}
with \(\vartheta(P,r,r')=\varphi(r,P-r,r',P-r')\). Thus, the reduced state is block-diagonal in the total momentum of subsystem \(12\). It is therefore weakly invariant under translations of \(12\). Therefore, tracing out part of a global physical state leads to a weakly invariant subsystem state.

This discussion suggests that, once one starts from the PN approach but wishes to describe a proper subsystem, a weakly invariant description emerges naturally at the level of the reduced state by application of the kinematical partial trace over subsystem 3. The appropriate framework therefore depends on the physical situation under consideration. For a closed system describing the totality of the relevant degrees of freedom, or the \enquote{whole universe}, the PN approach is appropriate. By contrast, when one discards part of a larger closed system, the kinematical partial trace does not generally produce a PN-physical state of the remaining subsystem; rather, it produces a weakly invariant state. 

At the level of trace-class operators, it turns out that, for general compact groups, the kinematical partial trace over subsystem $3$ has as its image the full ideal of weakly invariant trace-class operators of subsystem $12$.\footnote{The surjectivity claim underlying this statement can be regarded as the third-particle analogue of Ref.~\cite[Appendix: Compatibility with an external zero-charge state]{CastroRuiz2025RelativeSubsystems}, here formulated in terms of trace-class operators and proven by a different argument.} We capture this in the following proposition, proven in App.~\ref{app:proof_proposition}.
\begin{proposition} \label{prop:trace_surjectivity}
Let $G$ be compact and let $\mathcal B_1(\mathcal H_{123}^{\rm phys})$ denote the trace-class ideal on the physical Hilbert space of the composite system $123$. Then, after embedding physical trace-class operators into the kinematical Hilbert space, the kinematical partial trace over subsystem $3$ satisfies
\begin{equation} \label{eq:traccia-sul-fisico-invariant}
\operatorname{Tr}_3\left[\mathcal B_1(\mathcal H_{123}^{\rm phys})\right]
=
\mathcal B_1(\mathcal H_{12})^{G}.
\end{equation}
\end{proposition}

In this sense, starting from a PN description of a closed system, tracing out a subsystem naturally yields a reduced state in the invariant algebra, where the extra-particle degree of freedom becomes relevant. This suggests that the QI approach arises naturally as the subsystem description induced by the PN approach through subsystem discarding.

Importantly, the previous observation does not imply that the Third-Particle Paradox can be resolved for all states merely by using the ordinary partial trace as the notion of subsystem discarding after PN physicalisation. In fact, relational information about subsystem $12$ accessible from Eve's perspective may be washed away by the projector $\Pi_{123}$ followed by the standard partial trace. We capture this intuition via the following statistical condition.
\begin{definition}[PN--QI statistical consistency condition]\label{def:PN_vs_QI}
We say that $\operatorname{T}_{\rm PRT,QI}^{3|1}$ \emph{satisfies the PN--QI statistical consistency condition} for a kinematical state $\rho_{123} \in \mathcal B_1(\mathcal H_{123})$, with induced PN perspectival state $
\rho_{C,23|1}^{\rm PN} \;:=\; P_0^C\otimes \rho_{23|1}$,
if, for every observable $\mathcal A_{C,2|1} \in \mathcal{B}(\hilbert_{C,2|1})^G$ with corresponding weakly-invariant observable $\mathcal A_{12}^{\mathrm{inv}} \in \mathcal{B}(\mathcal{H}_{12})^G$, it holds that
\begin{equation} \label{eq:statistic_PN_vs_QI}
     \Tr\left[\left(\mathcal{A}_{12}^{\mathrm{inv}} \otimes \mathbb{I}_3\right) \rho_{123} \right] = \Tr\left[\mathcal{A}_{C,2|1} \operatorname{T}_{\rm PRT,QI}^{3|1}\!\left[\rho_{C,23|1}^{\rm PN}\right]\right]
\end{equation}
\end{definition}
We then characterise the domain of validity of the PN--QI consistency conditions via the following Theorem, proven in App.\nbs\ref{app:proof_theorem_3}.

\begin{theorem}[Characterisation of the PN--QI consistency condition]\label{thm:PN_QI}
    The PRT-QI satisfies the PN-QI statistical consistency condition if and only if the kinematical state $\rho_{123} \in \mathcal{B}_1(\hilbert_{123})$ is in the subspace
    \begin{equation}
    \Gamma \;=\;
    \Bigl\{ \; \rho_{123} \in \mathcal{B}_1(\hilbert_{123}) \;\bigm|\; \mathcal{G}^{(12)}\left[\Tr_3\left[\rho_{123}\right]\right] = \Tr_3\left[\Phi^{(123)}\left[\rho_{123}\right]\right]
    \Bigr\}.
    \end{equation}
\end{theorem}

As in the PN case, one may formulate a weaker, normalised version of the statistical condition\nbs\eqref{eq:statistic_PN_vs_QI}. This replaces the linear subspace \(\Gamma\) by the kinematical set
\begin{equation}
\widehat{\Gamma}
=
\Bigl\{
\rho_{123}\in\mathcal B_1(\mathcal H_{123})
\;\bigm|\;
\mathcal G^{(12)}
\!\left[
\Tr_3[\rho_{123}]
\right]
\sim
\Tr_3
\!\left[
\Phi^{(123)}[\rho_{123}]
\right]
\Bigr\}.
\end{equation}

The set $\widehat\Gamma$ consists precisely of those kinematical operators for which the weakly invariant content of Eve's reduced description of subsystem $12$ coincides with the reduced description obtained after strong physicalisation of the global system $123$. For states outside $\widehat\Gamma$, the strong-twirling followed by discarding subsystem $3$ may fail to preserve weakly invariant relational information about subsystem $12$ that was externally accessible.\footnote{This statement, and more precisely Theorem\nbs\ref{thm:PN_QI}, may appear to be in tension with Proposition\nbs\ref{prop:trace_surjectivity}. However, the latter is a surjectivity statement: every weakly invariant operator $12$ can be realised as the partial trace of some physical operator $123$. This does not imply that the kinematical $\rho_{123}$ is \emph{faithful} in the sense of Proposition\nbs\ref{prop:trace_surjectivity}. That is, Proposition\nbs\ref{prop:trace_surjectivity} does not imply that, for a given kinematical state, its weakly invariant reduction coincides with the reduction obtained after applying strong invariance to that state. The Paradox, and hence Eq.\nbs\eqref{eq:statistic_PN_vs_QI}, concerns this state-dependent comparison, which is why the resolution is restricted to the domain $\widehat\Gamma$.} We provide explicit examples of this mechanism in App.\nbs\ref{app:example_paradox_PN_QI}. The Third-Particle Paradox is thus not a genuine contradiction, but rather the consequence of comparing two different layers of description without keeping track of which information is externally accessible and which is internally accessible. By contrast, the analysis of \cref{sec:weak-twirling-PRT} shows that, if one adopts a weakly-invariant description already at the level of the total system $123$, the internally and externally accessible relational information about subsystem $12$ coincides. 

Therefore, we have different levels at which relational information about subsystem \(12\) can be represented, depending on the status of the total system \(123\). We summarise the discussion in the following table.
\begin{table}[h]
\centering
\renewcommand{\arraystretch}{1.2}
\begin{tabular}{|c|c|c|}
\hline
\textbf{Role of 12} & \textbf{Status of 123} & \textbf{Space encoding relational information about 12} \\
\hline
Closed total system & --- & \(\mathcal{B}_1(\mathcal H_{12}^{\rm phys})\) \\
\hline
Subsystem & Globally trivial charge & \(\Tr_3\,\mathcal{B}_1(\mathcal H_{123}^{\rm phys})= \mathcal{B}_1(\mathcal H_{12})^G\) \\
\hline
Subsystem & Charge-superselected & \(\mathcal{B}_1(\mathcal H_{12})^G\) \\
\hline
\end{tabular}
\caption{Operator spaces encoding relational information about subsystem \(12\), depending on the status of the total system. When \(12\) is itself regarded as a closed total system, relational information is encoded in the trivial-charge sector. If \(12\) is instead a subsystem of a physical system \(123\) with globally trivial charge, the charge of subsystem \(12\) compensates the charge of subsystem \(3\); after tracing out \(3\), the reduced description of \(12\) is weakly invariant. Finally, if the total system is described from the outset in the QI framework, relational information about \(12\) is encoded in the subsystem weakly invariant operator space.}
\revtexfloatlabel{tab:relational_info_levels}
\end{table}

These results clarify the physical meaning of the two QRF frameworks. The contrast between them is not a matter of one approach being more fundamental than the other, but of their stability under subsystem discarding: the PN approach consistently describes a closed system subject to a global constraint, whereas the QI approach remains stable when one restricts to a proper subsystem of a larger whole. Which framework is appropriate is therefore dictated by the physical question at hand, namely whether the relevant degrees of freedom are being treated as a closed system or as an open subsystem coupled to an environment. This is demonstrated by our independent analysis of the Third-Particle Paradox based on statistical consistency conditions across the different approaches. In particular, the impossibility of solving the Paradox in the PN approach outside \(\widehat\Lambda\) provides an operational and mathematical confirmation that the PN description is not, in general, stable under subsystem discarding.

Furthermore, our framework shows that the QI approach accommodates arbitrary subsystems, since the corresponding statistical consistency condition can be solved for all states. While Ref.~\cite{CastroRuiz2025RelativeSubsystems} argues that the QI framework is the appropriate one for arbitrary subsystems, our statistical conditions make this statement precise. Two cases must be distinguished, and our work separates them sharply. When the total system~123 is itself described within the QI approach, the statistical consistency condition holds for \emph{all} states (Theorem~\ref{thm:PRT_QI}): the weakly invariant information accessible from the external frame coincides with that accessible from the internal QRF after internal subsystem discarding. A description that adopts the weakly invariant algebra already at the level of the total system is therefore stable under discarding any of its parts. When, instead, the total system is described within the PN approach and the QI description of subsystem~12 is reached through the kinematical partial trace, this is no longer the case. The QI description does emerge in this way: tracing out particle~3 from a globally physical state lands subsystem~12 in the weakly invariant algebra, and this map is onto (Proposition~\ref{prop:trace_surjectivity}). Surjectivity, however, concerns which states are \emph{reachable}, not whether a \emph{given} state is faithfully reduced: for a fixed PN state $123$, the weakly invariant reduction of subsystem~12 need not retain the weakly invariant relational information that the external frame accesses by tracing out particle~3 directly. This is made precise in Theorem~\ref{thm:PN_QI}: the PN--QI consistency condition is solved only on the proper subset $\widehat\Gamma$, with an explicit product state outside it given in App.~\ref{app:example_paradox_PN_QI}. Hence the stability of the QI description under subsystem discarding holds in full when the weakly invariant algebra is adopted from the outset; recovering it by physicalising 123 and then tracing out particle 3 does not, in general, restore, state by state, all subsystem relational information that is accessible externally.

The Paradox of the Third Particle in the PN approach clearly recalls the argument introduced in~\cite{Rovelli2014WhyGauge}. In this work, the author argues that gauge-dependent variables are not just mathematical redundancies, but rather provide the handles through which subsystems couple. Indeed, the structural obstruction underlying point~\textit{(2)} of the Paradox provides a concrete QRF realisation of this mechanism. When subsystem $12$ is treated as isolated and one imposes strong invariance, only the 
trivial-charge sector $p_1 + p_2 = 0$ is retained. The non-trivial charge sectors of $12$ 
are projected out and discarded as pure gauge redundancy. However, once particle $3$ is 
added and one imposes strong invariance on the full system $123$, the physical Hilbert space 
contains, in general, sectors in which the charge of $12$ is non-trivial and correlated to 
that of $3$, via the constraint $p_1 + p_2 = -p_3$. These are precisely the sectors that 
were discarded when treating $12$ in isolation. In other words, the non-trivial charge 
sectors of $12$---which appear redundant from the perspective of $12$ alone---encode exactly 
how $12$ couples to its environment when the system is open. The theory obtained by imposing strong 
invariance separately on $12$ and $3$---that is, the factorised subspace selected by the RT---is, in general, strictly smaller than the theory 
obtained by first coupling $12$ and $3$ at the kinematical level and only afterwards imposing 
gauge invariance on the total system. The non-trivial charge sectors of 
$12$ that appear when $3$ is present are therefore not gauge artefacts, but carry genuine 
relational information about how $12$ interacts with its environment. In this sense, the impossibility of resolving the Third-Particle Paradox within the PN 
approach for states outside $\hat \Lambda$ is not merely a technical obstruction. It reflects that strong invariance, which implements the gauge constraint at the quantum level, by construction eliminates the degrees of freedom that become physically relevant when a subsystem interacts with an environment. Since this mechanism is not genuinely quantum, but due to the gauge structure of the theory, it is also manifested at the classical level. In App.~\ref{app:classical} we exhibit a classical instance of this mechanism, where imposing the classical constraint and then discarding particle~3 erases relational information about subsystem~12 that remains externally accessible. This complements the classical perspectival formulation of the Paradox recently established in Ref.~\cite[Sec.\nbs3]{rennerhausmann},\footnote{Ref.~\cite{rennerhausmann} appeared on the arXiv while this work was being completed. They establish a no-go theorem holding whenever physical systems, classical or quantum, are used as reference frames. Our findings are consistent with their discussion of the previous resolutions of the Paradox. In particular, our analysis makes transparent why the RT --- and, more generally, the description of a subsystem within the PN approach (Sec.\nbs\ref{sec:paradox-strong-twirling}) --- does not result in a consistent atlas of perspectives: discarding particles $3$ and $4$ simultaneously restricts to the sector selected by $\Pi_{12}\otimes\Pi_{34}$, thereby discarding the coupling between $12$ and $34$, whereas discarding them one at a time lands on $\Pi_{12}\otimes\Pi_{3}\otimes\Pi_{4}$ --- a strictly smaller sector, as in the latter the coupling between $3$ and $4$ is discarded as well. Furthermore, note that our PRT--QI\nbs\eqref{eq:PRT_QI_action} acts non-trivially on the extra-particle.} whose construction provides the classical counterpart of point \textit{(1)}. On the other hand, App.~\ref{app:classical} provides a classical instance of point~\textit{(2)} underlying the Third-Particle Paradox.  

This structure has a clear counterpart in the literature on entanglement entropy and edge modes in gauge theories\nbs\cite{Donnelly:2011hn,Donnelly2014nonabelian,Donnelly_2016,CasiniHuertaRosabal2014,Soni:2015yga,VanAcoleyen:2015ccp,Delcamp:2016eya,AraujoRegado:2025relational}, as mentioned in\nbs\cite{paradoxmuller}. With our findings, and more specifically Proposition\nbs\ref{prop:trace_surjectivity}, we can make this analogy precise. In the edge-modes literature, the physical Hilbert space does not factorise across a subregion $\Sigma$ and its complement $\bar\Sigma$, because Gauss' law ties the boundary charge of $\Sigma$ to that of $\bar\Sigma$---the direct analogue of the way strong twirling correlates the charge of subsystem $12$ with that of particle $3$ through $p_1+p_2=-p_3$ in our setting. In the extended-Hilbert-space approach, one embeds the physical state into an enlarged tensor-product Hilbert space and then traces out the complement~\cite{Donnelly:2011hn,Donnelly2014nonabelian,Donnelly_2016,Delcamp:2016eya}. In algebraic terms, the corresponding reduced description is block-diagonal in boundary-charge superselection sectors and is captured by the electric-center algebra~\cite{CasiniHuertaRosabal2014,Soni:2015yga}. Our Proposition~\ref{prop:trace_surjectivity} realises the same mechanism in the Third-Particle Paradox: embedding the physical states in the kinematical space and then tracing out particle $3$ from a globally physical state lands subsystem $12$ in the weakly invariant algebra $\mathcal B_1(\hilbert_{12})^G$, decomposed into charge sectors as in Eq.~\eqref{eq:states-traced-physical}. Equivalently, strong invariance followed by the partial trace acts as an incoherent twirl of subsystem \(12\). This is the minimal third-particle counterpart of the edge-mode mechanism. In Ref.~\cite{AraujoRegado:2025relational}, the electric center is recovered by incoherently twirling an extrinsic relational algebra over the electric corner symmetry group. In App.~\ref{app:edge_modes}, we show how to recover a distillable algebra for subsystem \(12\) in the Third-Particle Paradox, and then obtain the algebra of Eq.~\eqref{eq:states-traced-physical} via an incoherent twirl, in a scenario reminiscent of the extrinsic-frame construction of Ref.~\cite{AraujoRegado:2025relational}. As in those constructions, the resulting state lives in a charge-superselected algebra, and the corresponding entropy is not distillable. We also recover the third-particle analogue of the non-distillable, electric center entropy: the von Neumann entropy of Eq.~\eqref{eq:states-traced-physical} is given by $H(\{p_q\})+\sum_q p_q\big(\log d_q+S(\hat\tau^{(q)})\big)$ (see Corollary~\ref{cor:entropy_decomposition} in App.~\ref{app:edge_modes}) and is the third-particle counterpart of the three-term decomposition of Refs.~\cite{Donnelly:2011hn,Donnelly2014nonabelian} and the electric center entropy of Ref.~\cite{AraujoRegado:2025relational}. In this sense, the Third-Particle Paradox is a minimal toy model of the edge-mode mechanism: the non-trivial charge sectors of $12$, redundant when $12$ is treated in isolation, play the role of the boundary sectors through which $12$ couples to its environment when the system is open, in analogy to how the boundary charges at the boundary of $\Sigma$ describe the coupling between the subregion and its complement.

\section{Conclusions}\label{sec:conclusion}

In this work, we proposed a new resolution and contributed to unveiling the mechanism underlying the Paradox of the Third Particle~\cite{Angelo_2011} within the QRF setting. Our starting point was a careful diagnosis of why the Paradox emerges, which we traced back to two conceptually distinct mechanisms. The first is the QRF-covariance of the partial trace. We addressed this through the \emph{Perspective Trace}, which implements subsystem discarding consistently across QRF perspectives related by unitary transformations. The second mechanism is more subtle and specific to the PN framework: the projection onto the physical Hilbert space does not inherit the kinematical tensor-product structure, so that the partial traces taken before and after physicalisation are, in general, inequivalent operations. A natural candidate for discarding a subsystem at the physical Hilbert space level is the Relational Trace of~\cite{paradoxmuller}, which reformulates the discarding operation directly at the level of physical, i.e. relational, degrees of freedom. We argued, however, that satisfying its statistical consistency condition does not certify an operational resolution of the Paradox, showing an explicit product state on which the Relational Trace vanishes. Furthermore, we argued that the Relational Trace condition overlooks the mechanism at the heart of the Paradox: the clash between the partial trace performed before and after physicalisation.

Building on this diagnosis, we reformulated the Paradox via a statistical consistency condition and introduced the \emph{Perspective Relational Trace} (PRT), the notion of subsystem discarding that this condition singles out. A crucial feature of our approach is that we do not require the condition to hold for all states; in fact, we showed that our consistency condition cannot be satisfied for a certain class of states in the PN approach. This failure is not a defect but the main feature of our construction. The condition compares the reduced description obtained by tracing out particle~3 \emph{before} physicalisation with the one obtained \emph{after} physicalisation; the two disagree when the subsystem relational information externally accessible to Eve is not accessible internally once particle~3 is discarded after physicalisation. In this sense, we argued that our condition captures the operational content of the Paradox, whereas the Relational Trace condition, being evaluated entirely after physicalisation, overlooks this mismatch. Within the PN approach, this confines the resolution of the Paradox to a subset of states---including product states for which particle~3 is uncorrelated with subsystem~12. In the QI approach, by contrast, the analogous condition is satisfied for all states, since the subsystem relational information---encoded here in the weakly invariant algebra---accessible externally, to Eve, and internally, in the perspective of particle~1 after applying the PRT--QI, coincides. This rests on treating the full $G$-invariant algebra as relational, comprising relational observables together with the extra-particle degree of freedom; discarding the extra-particle would reinstate the Paradox within the QI framework as well.

Our technical findings have a transparent physical motivation. The PN approach imposes a global constraint on the whole system, and is therefore the appropriate description for a closed, isolated system. The non-trivial charge sectors that strong invariance discards as pure gauge are, however, precisely those through which a subsystem couples to its environment once the system is open\nbs\cite{Rovelli2014WhyGauge}. Imposing gauge invariance and then discarding a subsystem is not the same as discarding the subsystem and then imposing invariance: the latter misses the charges through which the subsystem couples to its environment. Thus, the charge-coupling obstruction underlying point \textit{(2)} of the Paradox is due to the gauge structure of the theory, rather than to genuinely quantum effects, complementing the classical perspectival version of the Paradox recently discussed in\nbs\cite{rennerhausmann}.  Furthermore, we showed that the QI approach is stable under subsystem discarding and can accommodate arbitrary subsystems, in line with the discussion in\nbs\cite{CastroRuiz2025RelativeSubsystems}. We make this statement precise through operational consistency conditions, and thereby separate three distinct levels at which subsystem~12 can be described once particle~3 is discarded.

If subsystem~12 is described within the PN approach (PN total system, PN subsystem), the Paradox cannot be resolved for all states: the consistency condition fails outside $\widehat\Lambda$, and the failure already occurs on product states for which particle~3 is uncorrelated with subsystem~12 (Theorem~\ref{thm:PRT_PN}). If, instead, the QI description is adopted from the outset (QI total system, QI subsystem), the condition holds for \emph{all} states: the relational information about~12 accessible externally, to Eve, and internally, after discarding particle~3, coincides (Theorem~\ref{thm:PRT_QI}). The intermediate case is a hybrid scenario: the total system is described in the PN approach and subsystem~12 is reached through the kinematical partial trace (PN total system, QI subsystem). At the level of operator algebras the QI description emerges in full---tracing out particle~3 maps the physical ideal onto the entire weakly invariant algebra (Proposition~\ref{prop:trace_surjectivity}). State by state, however, this is not the case: the externally accessible relational information is preserved only on the proper subset $\widehat\Gamma$ (Theorem~\ref{thm:PN_QI}), so that an instance of the Paradox persists. The stability of the QI description under subsystem discarding therefore holds in full only when it is adopted already for the total system; recovering it a posteriori from a globally physical PN state does not, by itself, restore all externally accessible information state by state.

Furthermore, by Proposition~\ref{prop:trace_surjectivity}, tracing out particle~3 from a globally physical state lands subsystem~12 in the weakly invariant algebra, and we argued that the QI description of a subsystem emerges from a PN description of the closed total system through the kinematical partial trace. This mechanism connects to the edge-mode program\nbs\cite{Donnelly:2011hn,Donnelly2014nonabelian,Donnelly_2016,CasiniHuertaRosabal2014,Soni:2015yga,VanAcoleyen:2015ccp,Delcamp:2016eya,Carrozza2024EdgeModesDynamicalFrames,AraujoRegado2025SoftEdges,AraujoRegado:2025relational}: the boundary charges that appear once particle~3 is traced out are the analogue, in our minimal setting, of those carried by the boundary of a spacetime subregion in a gauge theory, and they are genuine physical content rather than gauge redundancy. In this sense, the Third-Particle Paradox provides a toy model of the edge-mode mechanism.

Ultimately, our work shows that the Paradox of the Third Particle is not a genuine contradiction, but the consequence of comparing different levels of description without keeping track of which subsystem relational information is externally accessible and which internally.

\section*{Acknowledgments}

We thank Augustin Vanrietvelde and Guilhem Doat for discussions. We also thank Bruna Sahdo and Esteban Castro-Ruiz for feedback on a draft of this paper. A.P. further thanks Anne-Catherine de la Hamette for discussions on edge modes. L.A. is supported by the STeP2 grant (ANR-22-EXES-0013) of Agence Nationale de la Recherche (ANR) and the ANR grant TaQC (ANR-22-CE47-0012). 

\par\medskip
\noindent\textbf{Note.}\quad
Before the completion of this work we became aware of an upcoming work by Bruna Sahdo and
Esteban Castro-Ruiz\nbs\cite{sahdocastroruiz}, which analyses the compositional structure of
the extra-particle framework for quantum reference frame transformations, including the
paradox of the third particle.

\bibliographystyle{IEEEtran}
\bibliography{bibliography/bibliography}

\appendix
\renewcommand{\thesubsection}{\thesection.\arabic{subsection}}
\section{Perspective Neutral Approach: mathematical details}\label{app:pn_details}

We use the trace pairing between trace-class and bounded operators,
\begin{equation}
    \langle \rho,\mathcal{A}\rangle_{\rm tr}:=\Tr[\rho\mathcal{A}],
    \qquad \rho\in\mathcal{B}_1(\hilbert),\quad \mathcal{A}\in\mathcal{B}(\hilbert).
\end{equation}
For any $\mathcal{A} \in  \mathcal{B}(\hilbert)$ and $\rho \in  \mathcal{B}_1(\hilbert)$, let
$\mathcal{A}_{\rm phys}:=\Pi\mathcal{A}\Pi \in \mathcal{B}(\hilbert^{\rm phys})$ and
$\rho_{\rm phys}:=\Pi\rho\Pi \in  \mathcal{B}_1(\hilbert^{\rm phys})$, where both compressed operators can be equivalently represented on the kinematical or physical space. The trace pairing then satisfies
\begin{align} 
    \Tr_{\hilbert^{\rm phys}}\left[\rho_{\rm phys}\mathcal{A}_{\rm phys}\right]
    = & \  \Tr_{\hilbert}\left[\rho_{\rm phys}\mathcal{A}\right]
    = \Tr_{\hilbert}\left[\rho\mathcal{A}_{\rm phys}\right]. \label{eq:innerproductpn}
\end{align}
This coincides with the definition of the inner product in the physical Hilbert space as given in \cite{Vanrietvelde2020ChangePerspective}.\footnote{For non-compact groups, the trace pairing with both the state and the observable projected, may be ill-defined, since $\Pi$ does not, in general, define a bounded projector on the kinematical Hilbert space.} Note also that the compression map $\Phi[T]=\Pi T\Pi$ is self-adjoint with respect to this trace pairing.
Since, for compact unimodular $G$, the strong twirling map $\Pi$ is a bounded orthogonal projector with $\|\Pi\|\leq 1$, the projection of any bounded observable, $\mathcal{A}_{\mathrm{phys}}=\Phi[\mathcal{A}]$, remains bounded:
\begin{equation}
    \|\mathcal{A}_{\rm phys}\| = \|\Phi[\mathcal{A}]\| \leq \|\Pi\|^2\,\|\mathcal{A}\| \leq \|\mathcal{A}\|.
\end{equation}
Here we used the submultiplicativity of the operator norm. Since $\mathcal{R}_{S|i}^{(i)}(g)$ is unitary from $\hilbert^{\rm phys}$ to $\hilbert_{S|i}$ for complete ideal frames, the corresponding observable in the QRF of particle $i$ is again bounded. Thus, boundedness at the kinematical level automatically implies boundedness of the corresponding physical and QRF-relative observables. An analogous statement holds for states: if $\rho\in\mathcal{B}_1(\hilbert)$ is trace-class, then both its compression $\Phi[\rho]$ and its QRF-relative description remain trace-class. Indeed, products of bounded and trace-class operators are trace-class, and unitary/isometric conjugation preserves trace-classness. Accordingly, throughout the paper we assume all observables to be bounded and all states to be trace-class operators, independently of the Hilbert space under consideration. This ensures that the trace-pairing is always well defined.

For complete ideal internal frames, the QRF transformations in the PN approach coincide with those of the perspectival approach\nbs\cite{Giacomini2019,deLaHamette2020QRF}. In the latter, the relative orientations between all frames are assumed to be perfectly accessible from all perspectives and therefore QRF transformations are implemented by unitary operators. Under this assumption, the unitary change of perspective from particle $i$ to particle $j$, namely $U_{i\rightarrow j}:\hilbert_{S|i}\rightarrow\hilbert_{S|j}$, is given by \cite{deLaHamette2020QRF}
\begin{equation} \label{eq:unitary_perspectival}
    U_{i\rightarrow j}
    =
    \mathrm{SWAP}_{i,j} \circ
    \left(
    \mathbb{I}_i \otimes\int_G dg \ket{g^{-1}}\bra{g}_j \otimes U_{S\setminus j|i}^\dagger(g)
    \right).
\end{equation}
Here $S\setminus j|i$ denotes the systems other than $j$, described relative to particle $i$, and $\mathrm{SWAP}_{i,j}$ denotes the ordinary swap operator between particle $i$ and $j$, i.e.
\begin{equation}
	    \mathrm{SWAP}_{i,j} \ket{g}_i \otimes \ket{g'}_j = \ket{g'}_i\otimes \ket{g}_j.
\end{equation}

\section{Quantum Information Approach: mathematical details}\label{app:qi_details}

We now record the analogue, for the QI framework, of the trace-pairing identity discussed in App.\nbs\ref{app:pn_details}. Recall that the weak twirling map projects bounded operators onto the invariant algebra $\mathcal{B}(\hilbert)^G$. For any $\mathcal{A}\in\mathcal{B}(\hilbert)$ and $\rho\in\mathcal{B}_1(\hilbert)$, let $
    \mathcal{A}_{\rm inv}:=\mathcal{G}[\mathcal{A}]\in\mathcal{B}(\hilbert)^G$ and 
    $\rho_{\rm inv}:=\mathcal{G}[\rho]\in\mathcal{B}_1(\hilbert)^G$. 
Then the trace pairing satisfies
\begin{equation} 
    \Tr\left[\rho_{\rm inv}\mathcal{A}\right]
    =
    \Tr\left[\rho_{\rm inv}\mathcal{A}_{\rm inv}\right]
    =
    \Tr\left[\rho\mathcal{A}_{\rm inv}\right].
    \label{eq:HSproductQI}
\end{equation}
This follows from the fact that, for compact $G$, the weak twirling map $\mathcal{G}$ is an idempotent map onto $\mathcal{B}(\hilbert)^G$ and is self-adjoint with respect to the trace pairing. This is the QI analogue of Eq.\nbs\eqref{eq:innerproductpn}.

We now discuss the additional degree of freedom corresponding to the extra-particle\nbs\cite{CastroRuiz2025RelativeSubsystems}. At the algebraic level, the algebra of relational observables of $S|i$ reads
\begin{equation}\label{eq:real_rel_obs}
\mathcal O_{S|i}
:=
\left\{
\int_G dg \ |g\rangle\langle g|_i \otimes U_{S|i}(g)\,\mathcal A_{S|i}\,U_{S|i}^\dagger(g)
\;\middle|\;
\mathcal A_{S|i}\in\mathcal B(\mathcal H_{S|i})
\right\}.
\end{equation}
This algebra is contained in, but in general does not exhaust, the full weakly invariant algebra $\mathcal{B}(\hilbert)^G=\operatorname{Im}\mathcal{G}$. Indeed, since the left-regular representation acts only on the left irrep index of the reference frame, any operator of the form
\begin{equation}
    \mathcal{A}
    =
    \bigoplus_{q \in \hat G}\mathbb{I}_{i_L}^{(q)} \otimes \mathcal{A}_{i_R}^{(q)} \otimes \mathbb{I}_{S|i}
\end{equation}
is also weakly invariant. After the refactorisation associated with particle $i$, the full invariant algebra decomposes as
\begin{equation}
    \mathcal{B}(\hilbert_{C,S|i})^G
    =
    \bigoplus_{q\in\hat G}
    \mathbb{I}_{C_L}^{(q)}
    \otimes
    \mathcal{B}\!\left(\hilbert_{C_R}^{(q)}\otimes\hilbert_{S|i}\right).
\end{equation}
In the refactorised representation associated with particle $i$, the full invariant algebra factorises as
\[
    \mathcal{B}(\hilbert_{C,S|i})^G
    \cong
    \mathcal O_{S|i} \otimes \overline{\mathcal O}_{S|i},
\]
where $\mathcal O_{S|i}$ is the algebra of relational observables of $S|i$, while $\overline{\mathcal O}_{S|i}$ is the complementary algebra carried by the right-regular factor $C_R$, which is identified with the extra-particle subsystem.

Finally, we note that the $G$-twirl preserves both boundedness and trace-classness. Indeed, if $\mathcal{A}\in\mathcal{B}(\mathcal{H})$, then
\begin{align}
    \left\|\mathcal{G}\left[\mathcal{A}\right]\right\| \leq  \divG \int_G dg \ \left\|U(g)\mathcal{A} U(g)^\dagger \right\| = \|\mathcal{A}\|,
\end{align}
where we used unitary invariance of the operator norm. Similarly, for $\rho\in\mathcal{B}_1(\mathcal{H})$,
\begin{align}
    \left\|\mathcal{G}\left[\rho\right]\right\|_1 \leq  \divG \int_G dg \ \left\|U(g)\rho U(g)^\dagger \right\|_1 = \|\rho\|_1,
\end{align}
so $\mathcal{G}$ maps trace-class operators to trace-class operators. On states, the same map is CPTP. Moreover, unitary conjugation preserves both boundedness and trace-classness. Consequently, if observables are bounded and states are trace-class at the kinematical level, the corresponding invariant and QRF-relative operators retain these properties.

\section{Strong twirling decomposition for general compact groups}\label{app:strong_twirling}

Let us consider the composite system $\hilbert_{123}=\hilbert_{12}\otimes \hilbert_3$. Each subsystem Hilbert space $\hilbert_i$ decomposes as in Eq.\nbs\eqref{eq:HS_decomposition}. After recoupling subsystems $1$ and $2$, one obtains $ 
\hilbert_{12}\cong \bigoplus_{q_{12}\in\hat G}\hilbert_{12_L}^{(q_{12})}\otimes \hilbert_{12_R}^{(q_{12})}$. Hence
\begin{equation}
\hilbert_{123}
\cong
\bigoplus_{q_{12},q_3\in\hat G}
\Bigl(\hilbert_{12_L}^{(q_{12})}\otimes \hilbert_{3_L}^{(q_3)}\Bigr)
\otimes
\Bigl(\hilbert_{12_R}^{(q_{12})}\otimes \hilbert_{3_R}^{(q_3)}\Bigr).
\end{equation}

In this decomposition, the global action takes the form
\begin{equation}
U_{12}(g)\otimes U_3(g)
\cong
\bigoplus_{q_{12},q_3\in\hat G}
\Bigl(U_{12}^{(q_{12})}(g)\otimes U_3^{(q_3)}(g)\Bigr)
\otimes
\mathbb I_{12_R}^{(q_{12})}\otimes \mathbb I_{3_R}^{(q_3)}.
\end{equation}
Therefore, the strong twirling projector is
\begin{equation}
\Pi_{123}
:=
\frac{1}{\text{Vol}(G)}\int_G dg\,\bigl(U_{12}(g)\otimes U_3(g)\bigr)
\cong 
\bigoplus_{q_{12},q_3\in\hat G}
P_{(q_{12},q_3)}
\otimes
\mathbb I_{12_R}^{(q_{12})}\otimes \mathbb I_{3_R}^{(q_3)},
\end{equation}
where
\begin{equation}
P_{(q_{12},q_3)}
:=
\frac{1}{\text{Vol}(G)}\int_G dg\,
\Bigl(U_{12}^{(q_{12})}(g)\otimes U_3^{(q_3)}(g)\Bigr)
\end{equation}
is the orthogonal projector onto the invariant subspace of
$\hilbert_{12_L}^{(q_{12})}\otimes \hilbert_{3_L}^{(q_3)}$.

From Ref.\nbs\cite{delahamette2021perspectiveneutral} it follows that this invariant subspace is at most one dimensional and non-trivial if and only if $
q_3\simeq \bar q_{12}$ namely if subsystem $3$ carries the dual irrep of the total irrep of subsystem $12$. Denoting by $\Omega_{q_{12}}$ the rank-one projector onto the strongly-invariant line in the subspace
$\hilbert_{12_L}^{(q_{12})}\otimes \hilbert_{3_L}^{(\bar q_{12})}$, one has $
P_{(q_{12},q_3)}
=
\delta_{q_3,\bar q_{12}}\,
\Omega_{q_{12}}$. Hence the strong twirling assumes the compact form
\begin{equation}\label{eq:strong-twirling-general-groups}
\Pi_{123}
\cong 
\bigoplus_{q_{12}\in\hat G}
\Omega_{q_{12}}
\otimes
\mathbb I_{12_R}^{(q_{12})}\otimes
\mathbb I_{3_R}^{(\bar q_{12})}.
\end{equation}
Thus the strong twirling projects onto those sectors in which subsystem $3$ carries the dual irrep of the total irrep of subsystem $12$, while acting trivially on the multiplicity spaces. Since for abelian groups the irreps are one dimensional, it follows that $P_{(q_{12},q_3)}
=
\delta_{q_3,\bar q_{12}}$ and then the abelian expression in Eq.\nbs\eqref{eq:strong_projector_charge_decomposition} follows.

\section{Examples of the Paradox in the PN--QI setting}\label{app:example_paradox_PN_QI}

We illustrate Theorem~\ref{thm:PN_QI} with two explicit examples. The first shows that the inclusion $\widehat\Gamma \subsetneq \mathcal B_1(\mathcal H_{123})$ is
strict already on product states.  Let $G = \mathbb Z_2 = \{e,r\}$ for simplicity, with $r^2 = e$, and define $U := U(r)$. We have that
\begin{equation}
  U\ket{+} = \ket{+}, \qquad U\ket{-} = -\ket{-},
\end{equation}
so that $\ket{+}$ and $\ket{-}$ span the trivial and alternating charge
sectors, respectively. The global strong-twirling projector and the weak twirl
on subsystem~12 read
\begin{equation}
  \Pi_{123} = \tfrac{1}{2}\!\left(\mathbb I + U_1 U_2 U_3\right),
  \qquad
  \mathcal G^{(12)}[T] = \tfrac{1}{2}\!\left(T + (U_1 U_2)\,T\,(U_1 U_2)\right),
\end{equation}
the latter being the dephasing channel across the two eigenspaces of $U_1 U_2$,
which has eigenvalue $+1$ on $\{\ket{++},\ket{--}\}$ and $-1$ on
$\{\ket{+-},\ket{-+}\}$. Define
\begin{equation}
  \ket{\chi_\theta}_{12} = \tfrac{1}{\sqrt2}\!\left(\ket{+-} + e^{i\theta}\ket{-+}\right),
  \qquad
  \ket{\eta_\theta}_{12} = \sqrt p\,\ket{++} + \sqrt{1-p}\,\ket{\chi_\theta},
  \qquad 0 < p < 1,
\end{equation}
and consider the product state
$\ket{\psi_\theta}_{123} = \ket{\eta_\theta}_{12}\otimes\ket{+}_3$. The component
$\ket{++}$ lies in the $+1$ eigenspace of $U_1 U_2$, whereas $\ket{\chi_\theta}$
lies entirely in the $-1$ eigenspace; thus $\ket{\eta_\theta}$ is a coherent
superposition across the two charge sectors of subsystem~12, and the relational
phase $\theta$ is carried by its nontrivial-charge component $\ket{\chi_\theta}$. Since particle~3 factorises,
$\Tr_3[\ket{\psi_\theta}\bra{\psi_\theta}] = \ket{\eta_\theta}\bra{\eta_\theta}_{12}$.
The weak twirl removes the coherences between the two eigenspaces of $U_1 U_2$,
leaving
\begin{equation}\label{eq:app_ext}
  \mathcal G^{(12)}\!\left[\Tr_3(\ket{\psi_\theta}\bra{\psi_\theta})\right]
  = p\,\ket{++}\bra{++} + (1-p)\,\ket{\chi_\theta}\bra{\chi_\theta}.
\end{equation}

On the other hand, since $\ket{\chi_\theta}_{12} \otimes \ket{+}_3$ lies in the alternating sector of $123$, we get
\begin{equation}
  \Pi_{123}\ket{\psi_\theta}_{123}
  = \sqrt p\,\ket{++}_{12}\ket{+}_3 .
\end{equation}
Thus, the state does not lie in $\widehat\Gamma$ and the phase $\theta$ is no longer accessible after physicalisation. The relational phase
$\theta$ is carried by the nontrivial-charge component $\ket{\chi_\theta}$ of
subsystem~12. From Eve's perspective this component survives the kinematical
partial trace, and $\theta$ remains externally accessible. Strong physicalisation of the total system, by contrast, ties the charge
of subsystem~12 to that of particle~3: with particle~3 prepared in the trivial
sector, the projection $\Pi_{123}$ retains only the trivial-charge component
$\ket{++}$ of subsystem~12, discarding the $\theta$-carrying sector as gauge. Note, however, that in the previous example the phase $\theta$ is washed out by the strong twirling. In general, however, it is the strong twirling followed by the partial trace that removes weakly invariant relational information accessible from Eve's  perspective, as shown in the subsequent example. 

Consider
\begin{equation}
\ket{\psi_\theta}_{123}
= \tfrac12\big(\ket{{+}{+}{+}} + e^{i\theta}\ket{{+}{-}{-}}
+ \ket{{+}{+}{-}} + e^{i\theta}\ket{{-}{-}{-}}\big),
\end{equation}
which is entangled. Since
$\Pi_{123}$ retains only the components $\ket{{+}{+}{+}}$ and
$\ket{{+}{-}{-}}$,
\begin{equation}
\Pi_{123}\ket{\psi_\theta}_{123}
= \tfrac12\big(\ket{{+}{+}{+}} + e^{i\theta}\ket{{+}{-}{-}}\big),
\end{equation}
so the physical state $\Phi^{(123)}[\ket{\psi_\theta}\bra{\psi_\theta}]$ retains the
relational phase $\theta$ in the coherence between $\ket{{+}{+}{+}}$ and
$\ket{{+}{-}{-}}$: the strong twirling alone does not erase it. The two
components above carry orthogonal states of particle $3$, so tracing it out
removes their coherence,
\begin{equation}
\mathrm{Tr}_3\big[\Phi^{(123)}[\ket{\psi_\theta}\bra{\psi_\theta}]\big]
= \tfrac14\big(\ket{{+}{+}}\bra{{+}{+}} + \ket{{+}{-}}\bra{{+}{-}}\big)_{12},
\end{equation}
which is $\theta$-independent. Here it is the strong twirling followed by the partial trace that destroys the relational phase. Discarding particle $3$ before
physicalisation and then applying the weak twirl on subsystem $12$ instead gives
\begin{equation}
\mathcal{G}^{(12)}\big[\mathrm{Tr}_3[\ket{\psi_\theta}\bra{\psi_\theta}]\big]
= \tfrac12\ket{{+}{+}}\bra{{+}{+}} + \tfrac14\ket{{+}{-}}\bra{{+}{-}}
+ \tfrac14\ket{{-}{-}}\bra{{-}{-}}
+ \tfrac14\big(e^{-i\theta}\ket{{+}{+}}\bra{{-}{-}} + \mathrm{h.c.}\big),
\end{equation}
which retains $\theta$ through the coherence between $\ket{{+}{+}}$ and
$\ket{{-}{-}}$, internal to the $+1$ eigenspace of $U_1U_2$ and hence preserved
by $\mathcal{G}^{(12)}$. The right-hand side of the consistency condition is
$\theta$-independent while the left-hand side carries $\theta$ and has support on
$\ket{{-}{-}}$ absent from the former; the two cannot be proportional, so
$\ket{\psi_\theta}\bra{\psi_\theta}_{123}\notin\widehat\Gamma$. The loss of the externally accessible phase $\theta$ is here produced by
the partial trace acting after physicalisation. We note that no kinematical product state can exhibit this mechanism. Indeed, for a kinematical product state $\rho_{12}\otimes\rho_3$, any $\theta$-carrying term that survives strong twirling but is removed by the subsequent partial trace must connect distinct charge sectors of subsystem $12$, and is therefore also annihilated by $\mathcal{G}^{(12)}$. Hence, within the pure-state setting considered here, initial $12|3$ entanglement is necessary for this mechanism to occur. This product-state obstruction is specific to the Abelian structure of the present $\mathbb Z_2$ example, where all irreducible carrier spaces are one-dimensional and the projection cannot convert local intra-sector carrier--multiplicity coherences of the two product factors into coherence between the multiplicity degrees of freedom of subsystems $12$ and $3$. For non-Abelian symmetry groups admitting higher-dimensional irreducible carrier spaces and suitable multiplicity spaces, one can construct product states in which strong twirling transfers such carrier--multiplicity coherence into correlations between the two multiplicity spaces, so that the relational information is not erased by strong twirling alone, but by the subsequent partial trace.

\section{Classicality of point \textit{(2)} underlying the Paradox}
\label{app:classical}

We exhibit a classical instance of the mechanism of point~\textit{(2)} of
Sec.~\ref{sec:paradox_emergence}, complementing the classicality of point\nbs\textit{(1)} recently shown in\nbs\cite{rennerhausmann}. We work with the translation group
acting on three particles on a line, with external phase-space
coordinates $(x_i,p_i)$, $i=1,2,3$, and constraint
$C = p_1+p_2+p_3$. Let $\theta \in \mathbb{R}$, $k \neq 0$,
$0<p<1$, and consider the product probability measure
\begin{equation}
P_\theta \;=\; \,\Big[\, p\,\delta(x_1)\delta(x_2-x_1)\,\delta(p_1)\,\delta(p_2)
\;+\;(1-p)\,\delta(x_1)\delta(x_2-x_1-\theta)\,\delta(p_1)\,\delta(p_2-k)\,\Big]
\;\;\delta(x_3-c)\,\delta(p_3).
\label{eq:classicalstate}
\end{equation}
Particle 3 is uncorrelated with subsystem 12 and sharp in the
trivial-charge sector $p_3=0$. The relational information on $\theta$ is
carried by the branch with non-trivial subsystem charge $p_1+p_2=k$. Marginalising first over $(x_3,p_3)$ leaves the first
factor of Eq.~\eqref{eq:classicalstate} unchanged. In the
translation-invariant variables $u := x_2-x_1$ and $(p_1,p_2)$, the
reduced distribution reads
$p\,\delta(u)\delta(p_1)\delta(p_2)
+(1-p)\,\delta(u-\theta)\delta(p_1)\delta(p_2-k)$,
whose $u$-marginal $p\,\delta(u)+(1-p)\,\delta(u-\theta)$ reveals
$\theta$. The relational information about subsystem 12 is thus
externally accessible, through relational observables, after discarding particle 3. Imposing the constraint amounts to conditioning
Eq.~\eqref{eq:classicalstate} on $\{C=0\}$ and pushing the result forward to
the quotient by the gauge flow generated by $C$. Under $P_\theta$ the total
charge takes only the values $C=0$ and $C=k$, with probabilities $p$ and
$1-p$; the conditioning event therefore carries positive weight $p>0$ and the
conditional measure is elementary. The second branch, supported at $C=k\neq0$,
is annihilated, and
\begin{equation}
P_\theta\big|_{C=0}
= \delta(x_1)\,\delta(x_2-x_1)\,\delta(p_1)\,\delta(p_2)\,
\delta(x_3-c)\,\delta(p_3),
\end{equation}
which is $\theta$-independent. Furthermore, the pushforward of this measure to the quotient by the gauge flow is clearly also $\theta$-independent. The family
$\{P_\theta\}_\theta$, distinguishable by invariant observables at the
external level after discarding particle 3, is mapped to a single
physical state with no information on $\theta$ by the constraint.

The loss of $\theta$ is produced entirely by the constraint, which ties the charge
of subsystem 12 to that of particle 3: with particle 3 sharp in the
trivial sector, only the trivial-charge branch of 12 survives, with no information on $\theta$. This is the
classical counterpart of the first example of
App.~\ref{app:example_paradox_PN_QI}. We thus provided an example showing that the
charge-coupling obstruction underlying point~\textit{(2)} of the Paradox is
classical. Note that this example is a classical manifestation of the reduction to the subset $\hat \Gamma$ in our PN--QI condition (Theorem\nbs\ref{thm:PN_QI}). The classical manifestation of the reduction to the subset $\hat \Lambda$ in the PN--PN condition (Theorem\nbs\ref{thm:PRT_PN}) is even clearer: let $q \neq 0$ and consider the probability measure
\begin{equation} \label{eq:RT-counterpart-classical}
    \Tilde{P}_\theta = Q_{12}(\theta)\, \delta (p_3-q) \delta(x_3)
\end{equation}
where $Q_{12}(\theta)$ is a probability measure with support both on the
constraint surface $C_{12}=p_1+p_2=0$, where we assume the information about
$\theta$ to be encoded (as in Eq.~\eqref{eq:state-eve-2-paradox}), and on the sector
$p_1+p_2=-q$, where it has an atom of positive weight. Discarding $3$ and then imposing the constraint, information on $\theta$ is clearly preserved. However, imposing the factorised constraint selected by the RT would impose $C_{12}=p_1+p_2 = 0$ and $C_{3}=p_3 = 0$, so that the corresponding restriction of the measure normalised on the global constraint surface $C=0$ vanishes,
\begin{equation}
     \Tilde{P}_\theta\big|_{C_{12},C_3=0} = 0.
\end{equation}
Eq.\nbs\eqref{eq:RT-counterpart-classical} can be seen as the classical counterpart of Eq.\nbs\eqref{eq:counterexample}: the reduced state of particle 3 in\nbs\eqref{eq:counterexample} does not have support in the trivial charge sector, so that it does not have support on the factorised sector $\Pi_{12} \otimes \Pi_3$, exactly as its classical version\nbs\eqref{eq:RT-counterpart-classical} does not have support on the constraint surface $C_{12},C_3=0$ (although having support on $C=0$). 

\section{Distillable algebra for subsystem 12}\label{app:edge_modes}
In this appendix, we show how to obtain a distillable algebra for subsystem $12$, mimicking, in minimal form, the extrinsic-frame construction of Ref.\nbs\cite{AraujoRegado:2025relational}. There, QRFs are built from degrees of freedom in the complement of the subregion (Wilson lines). In our setting, we shall thus assume that the complement of subsystem $12$, that is, the third-particle, itself contains an ideal frame. We thus replace at the kinematical level $\hilbert_3 \to \hilbert_{\tilde{3}} \otimes \hilbert_R$, where $\hilbert_R = L^2(G)$ transforms under the left-regular representation. Since the image of the reduction map is the full system space for complete ideal frames, from the perspective of the frame $R$ the Hilbert space factorises between subsystem $12$ and the complement modulo frame,
\begin{equation}
    \hilbert_{123|R} = \hilbert_{12} \otimes \hilbert_{\tilde{3}}.
\end{equation}
Therefore, the frame $R$ assigns to subsystem $12$ the algebra
\begin{equation}
    \mathcal{A}_{12|R} := \mathcal{B}(\hilbert_{12}) \otimes \mathbb{I}_{\tilde 3},
\end{equation}
which is a Type I von Neumann factor and thus admits distillable entropies. This is the analogue of the distillable entropy assignment for the subregion $\Sigma$ from the perspective of the extrinsic frame in Ref.\nbs\cite{AraujoRegado:2025relational}. Note that, as in Ref.\nbs\cite{AraujoRegado:2025relational}, to obtain distillable entropies for $12$ we need to find a frame inside the third-particle and then we are able to obtain a gauge-invariant TPS between $12$ and the complement modulo frame, not the full complement of $12$. The frame $R$ admits a symmetry group $G_{\rm reor}$, given by the reorientations of the frame, implemented by the right-regular action $V_R(g)\ket{h}_R = \ket{hg^{-1}}_R$. Reorientations preserve the physical subspace since they commute with the global left action. The key observation is that the incoherent twirl over $G_{\rm reor}$ coincides, in the perspective of $R$, with the weak twirl $\mathcal G^{(12\tilde 3)}$. To recover the electric-center algebra and Eq.\nbs\eqref{eq:states-traced-physical}, we thus twirl over $G_{\rm reor}$, obtaining
\begin{equation}\label{eq:electric-center}
    \mathcal{G}_{\rm reor}(\mathcal{A}_{12|R})=\mathcal{G}^{(12\tilde 3)}(\mathcal{A}_{12|R}) = \mathcal{B}(\hilbert_{12})^G \otimes \mathbb{I}_{\tilde 3},
\end{equation}
which recovers the bounded counterpart of the charge-superselected space of Proposition\nbs\ref{prop:trace_surjectivity}. This is the third-particle analogue of the electric-center algebra\nbs\cite{Donnelly:2011hn,CasiniHuertaRosabal2014,Donnelly_2016,Donnelly2014nonabelian,Soni:2015yga,VanAcoleyen:2015ccp,Delcamp:2016eya,AraujoRegado:2025relational}.\footnote{More precisely, in\nbs\cite{AraujoRegado:2025relational} the electric-center algebra is defined as the pull-back to the PN structure of Eq.\nbs\eqref{eq:electric-center}. The two are clearly isomorphic under conjugation by the reduction map $\mathcal R^{(R)}_{12\tilde 3}$.}

We conclude by quantifying the entropic content of the charge-superselected reduced states.
\begin{corollary}[Entropy decomposition]\label{cor:entropy_decomposition}
Let $\rho_{12}=\Tr_{\tilde 3 R}\!\left[\Theta^{\rm k}\right]$ be as in Eq.\nbs\eqref{eq:states-traced-physical}, with $\Theta^{\rm k}$ a normalised state, and define
\begin{equation}
    p_q := \Tr\!\big[\tau^{(q)}_{12_R}\big],
    \qquad
    \hat\tau^{(q)} := \frac{\tau^{(q)}_{12_R}}{p_q} \quad \text{whenever } p_q\neq 0.
\end{equation}
Then $\{p_q\}_{q\in\hat G}$ is a probability distribution over the charge sectors of subsystem $12$, $\hat\tau^{(q)}$ is the normalised conditional state on the corresponding multiplicity factor, and, whenever the relevant entropies are finite,
\begin{equation}\label{eq:entropy_decomposition}
    S_{\rm vN}(\rho_{12})
    \;=\;
    H(\{p_q\}) \;+\; \sum_{q\in\hat G} p_q\Big(\log d_q + S_{\rm vN}\big(\hat\tau^{(q)}\big)\Big),
\end{equation}
where $H(\{p_q\})=-\sum_q p_q\log p_q$ denotes the Shannon entropy.
\end{corollary}
\begin{proof}
Positivity of $\Theta^{\rm k}$ implies $\tau^{(q)}_{12_R}\geq 0$, and $\Tr\big[\tfrac{1}{d_q} I^{(q)}_{12_L}\otimes\tau^{(q)}_{12_R}\big]=p_q$, so that $\sum_q p_q=\Tr[\rho_{12}]=1$. Since the charge blocks in Eq.\nbs\eqref{eq:states-traced-physical} are mutually orthogonal,
\begin{equation}
    S_{\rm vN}(\rho_{12})
    = -\sum_q \Tr\Big[\tfrac{1}{d_q} I^{(q)}_{12_L}\otimes\tau^{(q)}_{12_R}
      \log\Big(\tfrac{1}{d_q} I^{(q)}_{12_L}\otimes\tau^{(q)}_{12_R}\Big)\Big]
    = \sum_q \Big(p_q\log d_q - \Tr\big[\tau^{(q)}_{12_R}\log\tau^{(q)}_{12_R}\big]\Big),
\end{equation}
and substituting $\tau^{(q)}_{12_R}=p_q\,\hat\tau^{(q)}$ yields the claim.
\end{proof}

The three terms in Eq.\nbs\eqref{eq:entropy_decomposition} mirror the decomposition of Refs.\nbs\cite{Donnelly:2011hn,Donnelly2014nonabelian}. Moreover, for a physical state with description $\rho_{12\tilde 3|R}$ in the perspective of $R$, let $\rho_{12|R}:=\Tr_{\tilde 3}\big[\rho_{12\tilde 3|R}\big]$ denote the reduced state of subsystem $12$ prior to the twirl. Since the incoherent twirl is a mixed-unitary channel, concavity of the von Neumann entropy gives
\begin{equation}
    S_{\rm vN}\big(\rho_{12|R}\big) \;\leq\; S_{\rm vN}\big(\mathcal G^{(12)}[\rho_{12|R}]\big).
\end{equation}
Thus, the distillable relational entropy assigned by the frame is upper bounded by the non-distillable electric-center entropy, the third-particle counterpart of the corresponding bound of Ref.\nbs\cite{AraujoRegado:2025relational}.

\section{Proofs}
\subsection{Proof of Lemma\nbs\ref{lem:lemma_1}}\label{app:proof_lemma_1}
We start from
\begin{equation} 
\Tr\left[\mathcal{A}_{12}^{\mathrm{phys}}\Tr_3\left[\rho_{123} \right] \right]  = \Tr\left[\mathcal{A}_{2|1} \operatorname{T}_{\rm PRT}^{3|1} \left[ \rho_{23|1}\right]\right]
\end{equation}
and using Eq.\nbs\eqref{eq:innerproductpn} we can write 
\begin{equation} 
\Tr\left[\mathcal{A}_{12}^{\mathrm{phys}}  \Phi^{(12)} \left[\Tr_3\left[\rho_{123} \right] \right]  \right] = \Tr\left[\mathcal{A}_{2|1} \operatorname{T}_{\rm PRT}^{3|1} \left[ \rho_{23|1}\right]\right] 
\end{equation}
writing $\rho_{12} = \Tr_3\left[\rho_{123}\right]$, $\rho_{12}^{\rm phys} = \Phi^{(12)} \left[\rho_{12}\right]$ and expanding $\mathcal{A}_{2|1}$ we get
\begin{equation}  \Tr\left[\mathcal{A}_{12}^{\mathrm{phys}} \rho_{12}^{\rm phys} \right] = \Tr\left[ \mathcal R_{2}^{(1)} \mathcal{A}_{12}^{\mathrm{phys}} \mathcal R_{2}^{(1)\dagger} \operatorname{T}_{\rm PRT}^{3|1} \left[ \rho_{23|1}\right]\right]
\end{equation}
and using the cyclicity of the trace 
\begin{equation}  \Tr\left[\mathcal{A}_{12}^{\mathrm{phys}} \rho_{12}^{\rm phys} \right] = \Tr\left[  \mathcal{A}_{12}^{\mathrm{phys}} \mathcal R_{2}^{(1)\dagger} \operatorname{T}_{\rm PRT}^{3|1} \left[ \rho_{23|1}\right]\mathcal R_{2}^{(1)}\right] .
\end{equation}
Since the equality holds for all bounded physical observables 
$\mathcal A_{12}^{\mathrm{phys}}\in \mathcal{B}(\mathcal{H}_{12}^{\mathrm{phys}})$, and both operators on the left- and right-hand side belong to $\mathcal{B}_1(\mathcal{H}_{12}^{\mathrm{phys}})$, the non-degeneracy of the trace\footnote{We use the standard non-degeneracy of the trace pairing between $\mathcal B(\mathcal H)$ and $\mathcal B_1(\mathcal H)$: if $X,Y\in\mathcal B_1(\mathcal H)$ satisfy $\Tr(AX)=\Tr(AY)$ for all $A\in\mathcal B(\mathcal H)$, then $X=Y$.} implies
\begin{equation}
    \rho_{12}^{\rm phys}
    \overset{\rm phys}{=}
    \mathcal R_{2}^{(1)\dagger} \operatorname{T}_{\rm PRT}^{3|1}\left[\rho_{23|1}\right]\mathcal R_{2}^{(1)}.
\end{equation}
where $\overset{\rm phys}{=}$ means equality on the physical space. Conjugating both sides by the reduction maps and expanding $\rho_{23|1}$ we get
\begin{align} \label{eq:condition1pn}
\rho_{2|1} = \operatorname{T}_{\rm PRT}^{3|1} [\mathcal R_{23}^{(1)} \Phi^{(123)} [\rho_{123}] \mathcal R_{23}^{(1)\dagger}]
\end{align}
and finally expanding the left-hand side
\begin{align} \label{eq:final eq lemma 1}
\mathcal R_{2}^{(1)} \Phi^{(12)} \left[\Tr_{3}\left[\rho_{123}\right]\right] \mathcal R_{2}^{(1)\dagger} &= 
\operatorname{T}_{\rm PRT}^{3|1} [\mathcal R_{23}^{(1)} \Phi^{(123)} \left[\rho_{123}\right] \mathcal R_{23}^{(1)\dagger}].
\end{align}
Conversely, assume Eq.~\eqref{eq:final eq lemma 1}. For any bounded physical observable $\mathcal A_{12}^{\rm phys}$, let
$\mathcal A_{2|1}=\mathcal R_{2}^{(1)} \mathcal A_{12}^{\rm phys}\mathcal R_{2}^{(1)\dagger}$.
Multiplying both sides of Eq.~\eqref{eq:final eq lemma 1} by \(\mathcal A_{2|1}\), taking the trace, and using cyclicity together with Eq.\nbs\eqref{eq:innerproductpn}, one immediately recovers Eq.~\eqref{eq:statisticpn}. This proves the equivalence.
\qed

\subsection{Proof of Theorem~\ref{thm:PRT_PN}}\label{app:proof_theorem_1}

By Lemma\nbs\ref{lem:lemma_1}, the map $\operatorname{T}_{\rm PRT}^{3|1}$ satisfies the statistical consistency condition~\eqref{eq:statisticpn} for a kinematical state $\rho_{123}\in\mathcal B_1(\hilbert_{123})$ if and only if
\begin{equation}\label{eq:proof lemma1 condition}
\mathcal R_{2}^{(1)}\,\Phi^{(12)}\!\left[\Tr_3[\rho_{123}]\right]\mathcal R_{2}^{(1)\dagger}
=
\operatorname{T}_{\rm PRT}^{3|1}\!\left[\mathcal R_{23}^{(1)}\,\Phi^{(123)}[\rho_{123}]\,\mathcal R_{23}^{(1)\dagger}\right].
\end{equation}
We first restrict Eq.\nbs\eqref{eq:proof lemma1 condition} to physical states $\rho^{\rm phys}_{123}\in\mathcal B_1(\hilbert^{\rm phys}_{123})$. On physical states, Eq.\nbs\eqref{eq:proof lemma1 condition} fixes the form of the map. The crucial observation is that, since the reduction map is an isomorphism from the physical to the perspectival space, and the input of the PRT is perspectival, this restriction fixes the map at every point of its domain 
\begin{equation}\label{eq:proof explicit action}
\operatorname{T}_{\rm PRT}^{3|1}[\rho_{23|1}]
=
\mathcal R_{2}^{(1)}\,\Phi^{(12)}\!\left[\Tr_3\!\left[\mathcal R_{23}^{(1)\dagger}\rho_{23|1}\,\mathcal R_{23}^{(1)}\right]\right]\mathcal R_{2}^{(1)\dagger},
\qquad \rho_{23|1}\in\mathcal B_1(\hilbert_{23|1}).
\end{equation}
Therefore, this proves the existence and uniqueness of the map. Eq.\nbs\eqref{eq:proof explicit action} defines a CP and trace non-increasing map, being a composition of CPTP and CP trace non-increasing maps.\footnote{The partial trace $\Tr_3$ in Eq.\nbs\eqref{eq:proof explicit action} is understood via the kinematical embedding $J:\hilbert^{\rm phys}_{123}\to\hilbert_{123}$, with $J^\dagger J=\mathbb{I}_{\hilbert^{\rm phys}_{123}}$ and $JJ^\dagger=\Pi_{123}$, introduced in Sec.\nbs\ref{sec:paradox_emergence}, since $\hilbert^{\rm phys}_{123}$ does not inherit the kinematical tensor-product structure $\hilbert_{12}\otimes\hilbert_3$.} Then, the Kraus decomposition\nbs\eqref{eq:Kraus_PRT_PN} follows by expanding the partial trace and the projector map in their respective Kraus representations.

It remains to determine the kinematical domain at which the map\nbs\eqref{eq:proof explicit action} satisfies the consistency condition. We substitute $\rho_{23|1}=\mathcal R_{23}^{(1)}\Phi^{(123)}[\rho_{123}]\mathcal R_{23}^{(1)\dagger}$ into Eq.\nbs\eqref{eq:proof explicit action} to obtain
\begin{equation}
\operatorname{T}_{\rm PRT}^{3|1}\!\left[\mathcal R_{23}^{(1)}\Phi^{(123)}[\rho_{123}]\mathcal R_{23}^{(1)\dagger}\right]
=
\mathcal R_{2}^{(1)}\,\Phi^{(12)}\!\left[\Tr_3\!\left[\Phi^{(123)}[\rho_{123}]\right]\right]\mathcal R_{2}^{(1)\dagger}.
\end{equation}
Hence Eq.\nbs\eqref{eq:proof lemma1 condition} reads 
\begin{equation}
    \mathcal R_{2}^{(1)}\,\Phi^{(12)}[\Tr_3[\rho_{123}]]\,\mathcal R_{2}^{(1)\dagger}=\mathcal R_{2}^{(1)}\,\Phi^{(12)}[\Tr_3[\Phi^{(123)}[\rho_{123}]]]\,\mathcal R_{2}^{(1)\dagger}
\end{equation}
conjugating by the reduction map and using $\Phi^{(12)}\circ \Phi^{(123)} = \Phi^{(12)} \circ \Phi^{(3)}$ which follows from the invariance of the Haar measure, together with the cyclicity of the partial trace on the tensor factor $3$, we observe that the consistency condition Eq.\nbs\eqref{eq:statisticpn} holds if and only if 
\begin{equation}
\Phi^{(12)}\!\left[\Tr_3[\rho_{123}]\right]
\overset{\rm phys}{=}
\Phi^{(12)}\!\left[\Tr_3[\Pi_3\rho_{123}]\right].
\end{equation}
Then, we note that the previous equality on the physical space is satisfied if and only if it is satisfied on the kinematical space, that is, if and only if $\rho_{123}\in\Lambda$. This establishes that $\Lambda$ is the largest subspace containing the physical subspace for which the condition can be solved and concludes the proof.
\qed
\begin{remark}
The maximality of $\Lambda$ is to be understood among subspaces of $\mathcal B_1(\hilbert_{123})$ containing the physical subspace $\mathcal B_1(\hilbert^{\rm phys}_{123})$. The operational content of the Paradox lies in comparing the trace taken before physicalisation with the one taken after it. On physical states the strong twirling acts trivially, so the two traces coincide identically and therefore we require our condition to be satisfied there. Then, this requirement fixes the map to be the PRT. Imposing consistency on the physical sector therefore singles out the PRT and our characterisation captures the operational content of the Paradox.
\end{remark}

\subsection{Proof of Lemma\nbs\ref{lem:lemma_2}}\label{app:proof_lemma_2}

The proof is analogous to that of Lemma\nbs\ref{lem:lemma_1} in App.\nbs\ref{app:proof_lemma_1} with the strong twirling replaced by the weak twirling and the Schr\"{o}dinger reduction map replaced by the QI refactorisation map. The only non-trivial step is in the non-degeneracy of trace-pairing: the weakly invariant algebra is not of the form $\mathcal{B}(\hilbert')$ for some Hilbert space $\hilbert'$ and therefore we cannot apply the standard non-degeneracy argument. In this case, we can use the non-degeneracy of the pairing between a von Neumann algebra $M$ and its predual $M_{*}$, which, in our case, is the standard kinematical trace. Indeed, let $X,Y \in \mathcal{B}_1(\hilbert)^{G}$ such that $\Tr\left[X A\right] = \Tr\left[Y A\right]$ for all $A \in \mathcal{B}(\hilbert)^{G}$. Set $Z = X - Y$,  since the weakly invariant algebra is $*$-closed, we can choose $A = Z^\dagger$ hence we obtain $\Tr\left[Z Z^\dagger\right] = 0$ which forces $Z=0$. \qed 

\subsection{Proof of Theorem~\ref{thm:PRT_QI}}\label{app:proof_theorem_2}

By Lemma~\ref{lem:lemma_2}, the map $\operatorname{T}^{3|1}_{\rm PRT,QI}$ satisfies the statistical consistency condition~\eqref{eq:statistic_QI} for a kinematical state $\rho_{123}\in\mathcal B_1(\hilbert_{123})$ if and only if
\begin{equation}\label{eq:proof qi condition}
\mathcal V^{(2|1)}_1\!\left[\mathcal G^{(12)}[\Tr_3[\rho_{123}]]\right]
=
\operatorname{T}^{3|1}_{\rm PRT,QI}\!\left[\mathcal V^{(23|1)}_1[\mathcal G^{(123)}[\rho_{123}]]\right].
\end{equation}

The key identity is
\begin{equation}\label{eq:proof intertwine}
\Tr_3\!\left[\mathcal G^{(123)}[\rho_{123}]\right]=\mathcal G^{(12)}[\Tr_3[\rho_{123}]].
\end{equation}
Indeed, the partial trace is invariant under unitary conjugation on subsystem $3$, and therefore $\Tr_3\!\left[U_3(g) X U_3(g)^\dagger\right]=\Tr_3[X]$
for every operator $X$ and every $g\in G$.

Let $\sigma_{123} := \mathcal{G}^{(123)}[\rho_{123}] \in \mathcal{B}_1(\mathcal{H}_{123})^G$, which ranges over the whole domain since $\rho_{123}$ ranges over all trace-class operators and $\mathcal{B}_1(\mathcal{H}_{123})^G = \mathrm{Im}\big(\mathcal{G}^{(123)}\big)$. The induced state $\rho_{C,23|1}=\mathcal V^{(23|1)}_1[\sigma_{123}]$ ranges over the whole domain $\mathcal B_1(\hilbert_{C,23|1})^G$ as the QI factorisation map realises a $*$-isomorphism between the algebras. By Eq.~\eqref{eq:proof intertwine}, the left-hand side of Eq.~\eqref{eq:proof qi condition} depends on $\rho_{123}$ only through $\sigma_{123}$, hence Eq.~\eqref{eq:proof qi condition} is equivalent to the  identity
\begin{equation}\label{eq:proof qi explicit}
\operatorname  T^{3|1}_{\rm PRT,QI}[\rho_{C,23|1}]
=
\mathcal V^{(2|1)}_1\!\left[\Tr_3\!\left[(\mathcal V^{(23|1)}_1)^{-1}[\rho_{C,23|1}]\right]\right],
\qquad \rho_{C,23|1}\in\mathcal B_1(\hilbert_{C,23|1})^G.
\end{equation}
As the states $\rho_{C,23|1}$ exhaust the domain of the map, Eq.~\eqref{eq:proof qi explicit} determines the map at every point, proving existence and uniqueness. The PRT-QI is a composition of unitary conjugations with the partial trace, all of which are CPTP; therefore $\operatorname{T}^{3|1}_{\rm PRT,QI}$ is CPTP. Inserting the Kraus representation of $\Tr_3$ gives the Kraus decomposition~\eqref{eq:Kraus_PRT_QI}. Note that both sides of Eq.~\eqref{eq:proof qi condition} depend on $\rho_{123}$ only through its weak-invariant representation $\sigma_{123}$, and therefore the equivalence holds for \emph{every} $\rho_{123}\in\mathcal B_1(\hilbert_{123})$.
\qed

\subsection{Proof of Proposition\nbs\ref{prop:trace_surjectivity}}\label{app:proof_proposition}

We represent a generic physical operator $\Theta \in \mathcal{B}_1(\hilbert_{123}^{\rm phys})$ on the kinematical space as $
\Theta^{\rm k} := J \Theta J^\dagger \in \Pi_{123}\mathcal{B}_1(\hilbert_{123})\Pi_{123}$,
where $J$ denotes the canonical embedding of the physical Hilbert space into the kinematical Hilbert space. The physical projector takes the form Eq.\nbs\eqref{eq:strong-twirling-general-groups}  where $\Omega_q=\ket{\omega_q}\bra{\omega_q}$ is the rank-one projector onto the invariant line in
$\hilbert^{(q)}_{12,L}\otimes \hilbert^{(\bar q)}_{3,L}$ and $\ket{\omega_q}$ is the normalised maximally entangled invariant vector between the irrep $q$ carried by subsystem $12$ and the dual irrep $\bar q$ carried by subsystem $3$\nbs\cite[Lemmas~11 and~12]{delahamette2021perspectiveneutral}.

It follows that every $\Theta^{\rm k}\in \Pi_{123}\mathcal{B}_1(\hilbert_{123})\Pi_{123}$ admits the block form
\begin{equation}
\Theta^{\rm k}
=
\sum_{q,q'\in\hat G}
\ket{\omega_q}\bra{\omega_{q'}}_{12_L,3_L}
\otimes
T^{(q,q')}_{12_R,3_R},
\end{equation}
where $
T^{(q,q')}_{12_R,3_R}:
\hilbert^{(q')}_{12_R}\otimes \hilbert^{(\bar q')}_{3_R}
\longrightarrow
\hilbert^{(q)}_{12_R}\otimes \hilbert^{(\bar q)}_{3_R} $
are trace-class blocks. Taking the kinematical partial trace over subsystem $3$, the off-diagonal terms with $q\neq q'$ vanish, since different irrep spaces are orthogonal. For each $q$, the trace over $3_L$ gives $
\Tr_{3_L}\left[\Omega_q\right]
=
\frac{1}{d_q} I^{(q)}_{12_L}$. Defining $
\tau^{(q)}_{12_R}
:=
\Tr_{3_R}\left[T^{(q,q)}_{12_R,3_R}\right]
\in \mathcal B_1\left(\hilbert^{(q)}_{12_R}\right)$,
we obtain
\begin{equation}\label{eq:states-traced-physical}
\Tr_3[\Theta^{\rm k}]
=
\bigoplus_{q\in\hat G}
\frac{1}{d_q} I^{(q)}_{12_L}\otimes \tau^{(q)}_{12_R}.
\end{equation}
This is clearly an operator that is weakly-invariant under the left (gauge) action of the group. Since $T^{(q,q')}_{12_R,3_R}$ is arbitrary, so is $\tau^{(q)}_{12_R}$ and the claim is proven. \qed

\subsection{Proof of Theorem\nbs\ref{thm:PN_QI}}\label{app:proof_theorem_3}

Eq.\nbs\eqref{eq:statistic_PN_vs_QI} is equivalent to
\begin{equation}
\Tr \left[
\mathcal A^{\rm inv}_{12} \ \mathcal G^{(12)}[\Tr_3[\rho_{123}]]
\right]
=
\Tr\left[
\mathcal A^{\rm inv}_{12} \ 
\Tr_3\left[\Phi^{(123)}[\rho_{123}]\right]
\right].
\end{equation}
By Proposition\nbs\ref{prop:trace_surjectivity}, $\Tr_3\left[\Phi^{(123)}[\rho_{123}]\right ] \in \mathcal{B}_1(\hilbert_{12})^G$. Furthermore, using the non-degeneracy of the pairing between a von Neumann algebra and its predual (see App.\nbs\ref{app:proof_lemma_2}) it follows that the above condition holds for all $\mathcal{A
}^{\rm inv}_{12} \in \mathcal{B}(\hilbert_{12})^G$ if and only if
\begin{equation}
\mathcal G^{(12)}[\Tr_3[\rho_{123}]]
=
\Tr_3 \left[\Phi^{(123)}[\rho_{123}]\right],
\end{equation}
which proves the claim. \qed

\end{document}